\let\origvec\vec
\documentclass[smallcondensed,nospthms, 
]{svjour3}

\let\vec\origvec

\usepackage{verbatim}

\usepackage{pgfplots}
	\pgfplotsset{compat=1.12}

\usepackage{float}
\usepackage{amsthm} 
\usepackage{newtxtext,newtxmath} 
\usepackage[bookmarksdepth=3]{hyperref}
\usepackage[capitalize]{cleveref}
	\crefname{mtheorem}{Theorem}{Theorems}
	\crefname{equation}{}{}
	\crefname{figure}{Fig.}{}
	\crefformat{section}{\textsection#2#1#3}
	\crefrangeformat{equation}{(#3#1#4)-(#5#2#6)}
\usepackage{enumitem}
	\setenumerate{label = \emph{\textnormal{(\roman*)}}}
	\setitemize{label = $\bullet$}

\newtheorem{theorem}{Theorem}
\newtheorem{lemma}[theorem]{Lemma}

\newtheorem{mtheorem}{Theorem}

\theoremstyle{definition}
  \newtheorem{definition}[theorem]{Definition}
  \newtheorem{mdefinition}[mtheorem]{Definition}
  \newtheorem{notation}[theorem]{Notation}
  \newtheorem{remark}[theorem]{Remark}
  \newtheorem{example}[theorem]{Example}

\newcommand{\R}{\mathbb{R}}
\newcommand{\C}{\mathbb{C}}

\newcommand{\Z}{\mathbb{Z}}
\newcommand{\T}{\mathbb{T}}

\newcommand{\cH}{\mathcal{H}}

\renewcommand{\Re}{\mathrm{Re}\,}
\renewcommand{\Im}{\mathrm{Im}\,}
\newcommand{\ind}{\mathrm{ind}\,}
\newcommand{\sgn}{\mathrm{sgn}\,}
\newcommand{\ess}{\sigma_{\mathrm{ess}}}
\newcommand{\wn}{\mathrm{wn}}
\newcommand{\textbi}[1]{\textit{\textbf{#1}}}

\AtBeginDocument{%
  \setlength{\oddsidemargin}{\dimexpr(\paperwidth-\textwidth)/2-1in}%
  \setlength{\evensidemargin}{\oddsidemargin}%
  \setlength{\topmargin}{%
    \dimexpr(\paperheight-\textheight)/2-\headheight-\headsep-1in}%
}
\begin{document}

\title{The Witten Index for One-dimensional Non-unitary Quantum Walks with Gapless Time-evolution}

\author{Keisuke Asahara \and Daiju Funakawa \and Motoki Seki \and Yohei Tanaka}

\institute{
Keisuke Asahara \at
The Center for Data Science Education and Research, Shiga University,
1-1-1 Banba hikone, Shiga 522-8522 Japan \\
\email{keisuke-asahara@biwako.shiga-u.ac.jp} 
\and 
Daiju Funakawa \at
Department of Electronics and Information Engineering, Hokkai-Gakuen University,
Sapporo 062-8605, Japan \\
\email{funakawa@hgu.jp} 
\and 
Motoki Seki \at
Department of Mathematics, Faculty of Science, Hokkaido University
Kita 10, Nishi 8, Kita-Ku, Sapporo, Hokkaido, 060-0810, Japan \\
\email{seki@math.sci.hokudai.ac.jp} 
\and 
Yohei Tanaka \at
Division of Mathematics and Physics, Faculty of Engineering, Shinshu University, Wakasato, Nagano 380-8553, Japan \\
\email{20hs602a@shinshu-u.ac.jp}
}
\maketitle

\begin{abstract}
Recent developments in the index theory of discrete-time quantum walks allow us to assign a certain well-defined supersymmetric index to a pair of a unitary time-evolution $U$ and a $\mathbb{Z}_2$-grading operator $\varGamma$ satisfying the chiral symmetry condition $U^* = \varGamma U \varGamma.$ In this paper, this index theory will be extended to encompass non-unitary $U$. The existing literature for unitary $U$ makes use of the indispensable assumption that $U$ is essentially gapped; that is, we require that the essential spectrum of $U$ contains neither $-1$ nor $+1$ to define the associated index. It turns out that this assumption is no longer necessary, if the given time-evolution $U$ is non-unitary. As a concrete example, we shall consider a well-known non-unitary quantum walk model on the one-dimensional integer lattice, introduced by Mochizuki-Kim-Obuse.
\end{abstract}

\section{Introduction}
\label{section: introduction}

The theory of (discrete-time) quantum walks has attracted enormous attention over the past few decades. Despite its apparent simplicity, vast applications of this ubiquitous notion can be found across multiple disciplines. For instance, the physical utility of quantum walks is especially confirmed for quantum algorithms \cite{Grover-1996,Ambainis-Bach-Nayak-Vishwanath-Watrous-2001}, photosynthesis \cite{Mohseni-Rebentrost-Lloyd-Aspuru-Guzik-2008,Peruzzo-Lobino-Matthews-Matsuda-Politi-Poulios-Zhou-Lahini-Ismail-Worhoff-Bromberg-Silberberg-Thompson-OBrien-2010}, and topological insulators \cite{Kitagawa-Rudner-Berg-Demler-2010,Obuse-Kawakami-2011,Kitagawa-2012,Asboth-Obuse-2013}. The long-time limit of the velocity distribution of the quantum walker, known as the the weak limit theorem \cite{Konno-2002,Grimmett-Janson-Scudo-2004,Suzuki-2016}, has been a particularly active theme of rigorous mathematical research on quantum walks in the early years of the 21\textsuperscript{st} century. Other mathematical studies have taken various points of view: localisation \cite{Inui-Konishi-Konno-2004,Konno-2010,Segawa-2011,Cantero-Grunbaum-Moral-Velazquez-2012,Fuda-Funakawa-Suzuki-2017,Fuda-Funakawa-Suzuki-2018}, quantum walks on graphs \cite{Aharonov-Ambainis-Kempe-Vazirani-2001,Ambainis-2003,Portugal-2016}, non-linear analysis \cite{Maeda-Sasaki-Segawa-Suzuki-Suzuki-2018a,Maeda-Sasaki-Segawa-Suzuki-Suzuki-2018b,Maeda-Sasaki-Segawa-Suzuki-Suzuki-2019}, unitary equivalence classes \cite{Ohno-2016,Ohno-2017,Ohno-2018,Kuriki-Nirjhor-Ohno-2020}, time operators \cite{Sambou-Tiedra-2019,Funakawa-Matsuzawa-Sasaki-Suzuki-Teranishi-2020}, and continuous limit \cite{Maeda-Suzuki-2019}.

The present article is a continuation of rigorous mathematical studies of index theory for \textbi{chirally symmetric quantum walks} from the perspective of supersymmetric quantum mechanics \cite{Cedzich-Grunbaum-Stahl-Velazquez-Werner-Werner-2016,Cedzich-Geib-Grunbaum-Stahl-Velazquez-Werner-Werner-2018,Cedzich-Geib-Stahl-Velazquez-Werner-Werner-2018,Suzuki-2019,Suzuki-Tanaka-2019,Matsuzawa-2020,Cedzich-Geib-Werner-Werner-2020}. Such a quantum walk can be naturally identified with a pair of a time-evolution operator $U : \cH \to \cH$ and a unitary self-adjoint operator $\varGamma: \cH \to \cH,$ satisfying the chiral symmetry condition;
\begin{equation}
\label{equation: chiral symmetry}
U^* = \varGamma U \varGamma,
\end{equation}
where $\varGamma$ gives a $\Z_2$-grading of the underlying state Hilbert space $\cH = \ker(\varGamma - 1) \oplus \ker(\varGamma + 1).$ The existing literature mentioned above allows us to assign a certain well-defined Fredholm index, denoted by $\ind(\varGamma,U),$ to each abstract chirally symmetric quantum walk $(\varGamma, U).$ Note that this assignation of the Fredholm index requires $U$ to be both \textbi{essentially unitary} (i.e. $U$ is a unitary element in the Calkin $C^*$-algebra) and \textbi{essentially gapped} (i.e. the essential spectrum of $U,$ denoted by $\ess(U),$ contains neither $-1$ nor $+1$).

The present article extends this index theory to encompass all those time-evolutions $U$ which fail to be essentially unitary. As a concrete example, we shall explicitly construct such a time-evolution $U$ with the property that it is essentially gapless, yet the associated index is well-defined. To put this into context, let us consider the following time-evolution operator on the state Hilbert space $\cH := \ell^2(\Z, \C^2)$ of square-summable $\C^2$-valued sequences;
\begin{equation}
\label{equation: definition of evolution operator of MKO}
U_{\textnormal{mko}} := SG \Phi C_2 S G^{-1} \Phi C_1,
\end{equation} 
where the operators $S, G, \Phi, C_1, C_2$ are defined respectively as the following block-operator matrices with respect to the orthogonal decomposition $\cH = \ell^2(\Z, \C) \oplus \ell^2(\Z, \C);$
\[
S := 
\begin{pmatrix}
L & 0 \\
0 & L^{-1}
\end{pmatrix}, \quad 
G := 
\begin{pmatrix}
e^{\gamma} & 0 \\\centering
0 & e^{-\gamma(\cdot + 1)}
\end{pmatrix}, \quad 
\Phi := 
\begin{pmatrix}
e^{i \phi} & 0 \\
0 & e^{-i\phi(\cdot + 1)}
\end{pmatrix}, \quad 
C_j := 
\begin{pmatrix}
\cos \theta_j & i \sin \theta_j \\
i \sin \theta_j & \cos \theta_j 
\end{pmatrix},
\]
where $L$ is the unitary bilateral left-shift operator defined by $L \Psi := \Psi(\cdot + 1)$ for each $\Psi \in \ell^2(\Z, \C),$ and where we assume that four $\R$-valued sequences $\gamma = (\gamma(x))_{x \in \Z}$, $\phi = (\phi(x))_{x \in \Z}, \theta_1 = (\theta_1(x))_{x \in \Z}, \theta_2 = (\theta_2(x))_{x \in \Z},$ all of which are identified with the corresponding multiplication operators on $\ell^2(\Z, \C),$ admit the following two-sided limits:
\begin{equation}
\label{equation: existence of limits}
\xi(\star) := \lim_{x \to \star} \xi(x) \in \R, \qquad \xi \in \{\gamma, \phi, \theta_1, \theta_2\}, \quad \star = \pm \infty.
\end{equation}

This model is a natural generalisation of the homogenous model considered in \cite[\textsection III.A]{Mochizuki-Kim-Obuse-2016} with the time-evolution \cref{equation: definition of evolution operator of MKO} being consistent with the experimental setup in \cite{Regensburger-Bersch-Miri-Onishchukov-Christodoulides-Peschel-2012} (see \cite[\textsection I-II]{Mochizuki-Kim-Obuse-2016} for details). Note that $U_{\textnormal{mko}}$ is non-unitary, unless $\gamma$ is identically zero. We shall explicitly construct a $\Z_2$-grading operator $\varGamma_{\textnormal{mko}} : \cH \to \cH$ in a highly non-trivial fashion, so that $(\varGamma_{\textnormal{mko}},U_{\textnormal{mko}})$ forms a chirally symmetric quantum walk. Complete classification of the two topological invariants $\ind(\varGamma_{\textnormal{mko}},U_{\textnormal{mko}})$ and $\ess(U_{\textnormal{mko}})$ can be found in this paper. In particular, we show that $\ess(U_{\textnormal{mko}})$ is a subset of the union of the unit circle $\T$ and 
the real line $\R,$ given explicitly by the following formula;
\[
\ess(U_{\textnormal{mko}}) = \sigma(-\infty) \cup \sigma(+\infty),
\]
where the sets $\sigma(\pm \infty) \subseteq \T \cup \R$ depend only on the two asymptopic values $\theta_1(\pm \infty), \theta_2(\pm \infty).$ As in \cref{figure: three cases}, it is shown in this paper  that for each $\star = \pm \infty,$ there exists a well-defined subinterval $[\gamma_-(\star),\gamma_+(\star)]$ of $[0,\infty],$ which enables us to classify $\sigma(\star)$ into $6$ different cases in total, depending on the sign $s(\star)$ of $-\sin \theta_1(\star) \sin \theta_2(\star).$

\begin{figure}[H]
\centering
\begin{tikzpicture}
\begin{axis}[ticks=none, xmin=-2, xmax=2, ymin= -2, ymax=2, legend pos = north west, axis lines=center, xlabel=$\Re$, ylabel=$\Im$, xlabel style={anchor = north}
, width = 0.4\textwidth, height = 0.4\textwidth, clip=false
]
	\addplot [domain=0:2*pi,samples=50, smooth, dashed]({cos(deg(x))},{sin(deg(x))}); 
	\addplot [domain= 6*pi/10: 9*pi/10, samples=50, smooth, ultra thick, lightgray]({cos(deg(x))},{sin(deg(x))}); 
	\addplot [domain= -6*pi/10: -9*pi/10, samples=50, smooth, ultra thick, lightgray]({cos(deg(x))},{sin(deg(x))}); 
	\addplot [domain= 6*pi/10: 9*pi/10, samples=50, smooth, ultra thick]({-cos(deg(x))},{sin(deg(x))}); 
	\addplot [domain= -6*pi/10: -9*pi/10, samples=50, smooth, ultra thick]({-cos(deg(x))},{sin(deg(x))}); 
    \node at (0,2.5) {\textbf{Case I}};
    \node at (0,-2.5) {$|\gamma(\star)| \leq  \gamma_-(\star)$};
\end{axis}
\end{tikzpicture}
\begin{tikzpicture}
\begin{axis}[ticks=none,xmin=-2, xmax=2, ymin=-2, ymax=2, legend pos = north west, axis lines=center, xlabel=$\Re$, ylabel=$\Im$, xlabel style={anchor = north}
, width = 0.4\textwidth, height = 0.4\textwidth, clip=false
]
	\addplot [domain=0:2*pi,samples=50, smooth, dashed]({cos(deg(x))},{sin(deg(x))}); 
	\addplot [domain= 8*pi/10: pi, samples=50, smooth, ultra thick, lightgray]({cos(deg(x))},{sin(deg(x))}); 
	\addplot [domain= -8*pi/10: -pi, samples=50, smooth, ultra thick, lightgray]({cos(deg(x))},{sin(deg(x))}); 
	\addplot [domain= 0: 1, samples=50, smooth, ultra thick, lightgray]({-1.5*x - (1-x)/(1.5)},{0}); 
	\addplot [domain= 8*pi/10: pi, samples=50, smooth, ultra thick]({-cos(deg(x))},{sin(deg(x))}); 
	\addplot [domain= -8*pi/10: -pi, samples=50, smooth, ultra thick]({-cos(deg(x))},{sin(deg(x))}); 
	\addplot [domain= 0: 1, samples=50, smooth, ultra thick]({1.5*x + (1-x)/(1.5)},{0}); 
    \node at (0,2.5) {\textbf{Case II}};
    \node at (0,-2.5) {$\gamma_-(\star) < |\gamma(\star)| < \gamma_+(\star)$};
\end{axis}
\end{tikzpicture}
\begin{tikzpicture}
\begin{axis}[ticks=none,xmin=-2, xmax=2, ymin=-2, ymax=2, legend pos = north west, axis lines=center, xlabel=$\Re$, ylabel=$\Im$, xlabel style={anchor = north}
, width = 0.4\textwidth, height = 0.4\textwidth, clip=false
]
	\addplot [domain=0:2*pi, samples=50, smooth, dashed]({cos(deg(x))},{sin(deg(x))}); 
	\addplot [domain= 0: 1, samples=50, smooth, ultra thick, lightgray]({-1.8*x - (1.2)*(1-x)},{0}); 
	\addplot [domain= 0: 1, samples=50, smooth, ultra thick, lightgray]({-x/(1.2) - (1-x)/(1.8)},{0}); 
	\addplot [domain= 0: 1, samples=50, smooth, ultra thick]({1.8*x + (1.2)*(1-x)},{0}); 
	\addplot [domain= 0: 1, samples=50, smooth, ultra thick]({x/(1.2) + (1-x)/(1.8)},{0}); 
    \node at (0,2.5) {\textbf{Case III}};
    \node at (0,-2.5) {$\gamma_+(\star) \leq |\gamma(\star)|$};
\end{axis}
\end{tikzpicture}
\caption{
The set $\sigma(\star) \subseteq \T \cup \R$ is classified into Cases I, II, III as above according to the size of $|\gamma(\star)|.$ If $s(\star)  = 1$ (resp. if $s(\star)  = -1$), then the black regions (resp. gray regions) in each of the above three cases depict the subset $\sigma(\star).$ Therefore, there are $6$ distinct cases in total. In particular, $\sigma(\star)$ is a connected subset of $\T \cup \R$ containing either $-1$ or $+1,$ and so Case II is of significant importance.}
\label{figure: three cases}
\end{figure}
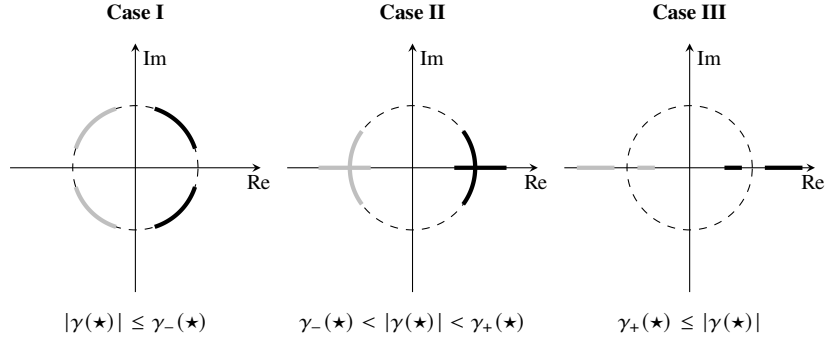
The present article is organised as follows. \cref{section: preliminaries} is a preliminary. In \cref{section: chirally symmetric quantum walks} we state the main theorems of this paper, proofs of which will be completely deferred to \cref{section: proof of main theorem}. In particular, it is shown in \cref{theorem: MKO model} that the $2$-step chirally symmetric quantum walk $(\varGamma_{\textnormal{mko}}, U_{\textnormal{mko}})$ can be naturally generalised to another $m$-step chirally symmetric quantum walk, denoted by $(\varGamma_m, U_m)$ in this paper, where $m$ can be any fixed non-zero integer. This new model also unifies the one-dimesional unitary quantum walks in \cite{Ambainis-Bach-Nayak-Vishwanath-Watrous-2001,Konno-2002,Suzuki-2016,Kitagawa-Rudner-Berg-Demler-2010,Kitagawa-Broome-Fedrizzi-Rudner-Berg-Kassal-Aspuru-Demler-White-2012,Kitagawa_2012,Fuda-Funakawa-Suzuki-2017,Fuda-Funakawa-Suzuki-2018,Fuda-Funakawa-Suzuki-2019,Suzuki-Tanaka-2019,Matsuzawa-2020,Tanaka-2020}. Complete classification of the two topological invariants $\ind(\varGamma_m, U_m)$ and $\ess(U_m)$ can be collectively found in \cref{maintheorem: generalised mko}. With the aid of \crefrange{theorem: MKO model}{maintheorem: generalised mko}, we show in \cref{section: gapless time-evolution} that the non-unitary time-evolution $U_{\textnormal{mko}}$ can have a well-defined index, yet it is essentially gapless (see \cref{example: second example} for details). This construction is based upon Case II in \cref{figure: three cases}. In \cref{section: proof of main theorem}, we prove the main theorems in \cref{section: main theorem and discussion}. In particular, we shall make use of an abstract form of the one-dimensional bulk-boundary correspondence to fully classify $\ind(\varGamma_m, U_m).$ The paper concludes with the summary and discussion in \cref{section: concluding remarks}.

\section{Main theorems and discussion}
\label{section: main theorem and discussion}

Proofs of the main theorems of the current section can be collectively found in \cref{section: proof of main theorem}.

\subsection{Preliminaries}
\label{section: preliminaries}

By operators we shall always mean everywhere-defined bounded operators between Hilbert spaces throughout this paper. An operator $X$ is said to be \textbi{Fredholm}, if $\ker X, \ker X^*$ are finite-dimensional and if $X$ has a closed range. If $X$ is Fredholm, then the \textbi{Fredholm index} of $X$ is defined by
\[
\ind X := \dim \ker X - \dim \ker X^*.
\]
If the domain and range of $X$ are identical, then the (Fredholm) \textbi{essential spectrum} of $X$ is defined by $\ess(X) := \{z \in \C \mid X - z \mbox{ is not Fredholm}\}.$ In particular, we call $X$ \textbi{essentially gapped}, if $-1,+1 \notin \ess(X)$ following  \cite{Cedzich-Geib-Grunbaum-Stahl-Velazquez-Werner-Werner-2018,Cedzich-Geib-Stahl-Velazquez-Werner-Werner-2018}. 

A \textbi{chiral pair} on $\cH$ is any pair $(\varGamma,U)$ of a unitary self-adjoint operator $\varGamma$ on $\cH$ and an operator $U$ on $\cH$ satisfying the chiral symmetry condition \cref{equation: chiral symmetry}. Note that $\varGamma$ gives a $\Z_2$-grading of the underlying Hilbert space $\cH = \ker(\varGamma - 1) \oplus \ker(\varGamma + 1),$ and that $\varGamma = 1 \oplus (-1)$ with respect to this orthogonal decomposition, where $1$ denotes the identity operator on a Hilbert space throughout this paper. The operator $U$ can be written as $U = R + iQ,$ where $R, Q$ are the real and imaginary parts of $U$ respectively. More precisely, $R,Q$ admit the following block-operator matrix representations: 
\begin{align}
\label{equation: representation of R and Q}
R =  
\begin{pmatrix}
R_1 & 0 \\
0 & R_2
\end{pmatrix}_{\ker(\varGamma - 1) \oplus \ker(\varGamma + 1)}, \qquad 
Q = \begin{pmatrix}
0 & Q_0^* \\
Q_0 & 0
\end{pmatrix}_{\ker(\varGamma - 1) \oplus \ker(\varGamma + 1)},
\end{align}
where the first equality (resp. second equality) follows from $[\varGamma, R] := \varGamma R - R \varGamma = 0$ (resp. from $\{\varGamma, Q\} := \varGamma Q + Q \varGamma = 0$).  

\begin{definition}
Let $(\varGamma, U)$ be a chiral pair on a Hilbert space $\cH,$ and let $Q$ be the imaginary part of $U$ given by the second equality in \cref{equation: representation of R and Q}. Then the chiral pair $(\varGamma, U)$ is said to be \textbi{Fredholm}, if $0 \notin \ess(Q)$ (or, equivalently, $Q_0$ is Fredholm). In this case, the \textbi{Witten index} of the Fredholm chiral pair $(\varGamma, U)$ is defined by 
$
\ind(\varGamma, U) := \ind Q_0.
$
\end{definition}

We shall make use of the following \textbi{unitary invariance} property of the Witten index throughout this paper;

\begin{lemma}[unitary invariance]
\label{lemma: unitary invariance of the witten index}
Let $(\varGamma,U), (\varGamma',U')$ be two chiral pairs on two Hilbert spaces $\cH,\cH'$ respectively. If  $(\varGamma,U), (\varGamma',U')$ are \textbi{unitarily equivalent} in the sense that $(\varGamma',U') = (\epsilon^* \varGamma \epsilon, \epsilon^* U \epsilon)$ for some unitary operator $\epsilon : \cH' \to \cH,$ then $(\varGamma, U)$ is Fredholm if and only if so is $(\varGamma', U').$ In this case, we have $\ind(\varGamma, U) = \ind(\varGamma', U').$
\end{lemma}

\subsection{Main theorems}
\label{section: chirally symmetric quantum walks}

We are now in a position to introduce the main model of the present article;

\begin{mdefinition}
\label{definition: Um}
Let $m$ be a fixed non-zero integer, and let $(\varGamma_m, U_m)$ be the pair of the following block-operator matrices with respect to $\ell^2(\Z, \C^2) = \ell^2(\Z) \oplus \ell^2(\Z):$
\begin{align}
\tag{A1}
\label{equation: definition of VarGamma m}
\varGamma_m &:=
\begin{pmatrix}
p & qL^{m} \\
L^{-m}q^*  & -p(\cdot - m)
\end{pmatrix}, \\
\tag{A2}
\label{equation: definition of Um}
U_m &:=
 \begin{pmatrix}
p & qL^{m} \\
L^{-m}q^*  & -p(\cdot - m)
\end{pmatrix}
\begin{pmatrix}
e^{-2 \gamma(\cdot + 1)}a & e^{\gamma - \gamma(\cdot + 1)}b^* \\
e^{\gamma - \gamma(\cdot + 1)}b & -e^{2 \gamma}a
\end{pmatrix}, 
\end{align}
where we assume that three convergent  $\R$-valued sequences $\gamma = (\gamma(x))_{x \in \Z}, p = (p(x))_{x \in \Z} , a = (a(x))_{x \in \Z}$ and two convergent $\C$-valued sequences $q =(q(x))_{x \in \Z}, b = (b(x))_{x \in \Z}$ satisfy the following conditions:
\begin{align}
\tag{A3}
\label{equation: condition of p and q}
&p(x)^2 + |q(x)|^2 = 1, & &x \in \Z, \\
\tag{A4}
\label{equation: condition of a and b}
&a(x)^2 + |b(x)|^2 = 1, & &x \in \Z,  \\
\tag{A5}
\label{equation: limits of sequences}
&\xi(\pm \infty) := \lim_{x \to \pm \infty} \xi(x), & & \xi \in \{\gamma, p, a, q, b\}, \\
\tag{A6}
\label{equation: limits of theta and theta prime}
&
\theta(\pm \infty) := 
\begin{cases}
\arg q(\pm \infty), & q(\pm \infty) \neq 0, \\
0,             & q(\pm \infty) = 0, 
\end{cases}
&
&\theta'(\pm \infty) := 
\begin{cases}
\arg b(\pm \infty), & b(\pm \infty) \neq 0, \\
0,             & b(\pm \infty) = 0, 
\end{cases}
\end{align}
where $\arg w$ of a non-zero complex number $w$ is uniquely defined by $w = e^{i \arg w}$ and $\arg w \in [0,2\pi).$
\end{mdefinition}

The pair $(\varGamma_m, U_m)$ introduced in \cref{definition: Um}, where $\varGamma_m$ is unitary self-adjoint by \cref{equation: condition of p and q}, turns out to be a chiral pair. Indeed, $U_m$ can be uniquely written as $U_m = \varGamma_m C,$ where $C$ is self-adjoint, and so
\[
U_m^* = (\varGamma_m C)^* = C^* \varGamma_m^*  = C \varGamma_m  = \varGamma_m^2 C \varGamma_m = \varGamma_m U_m \varGamma_m,
\]
where the second last equality follows from $\varGamma_m^2 = 1.$ The chiral pair $(\varGamma_m, U_m)$  unifies all of the following existing models on the one-dimensional integer lattice $\Z:$ 
\begin{itemize}
\item If $m = 1$ and if $\gamma$ is identically $0,$ then $U_1$ is the unitary evolution of a \textbi{split-step quantum walk model} considered in \cite{Kitagawa-Rudner-Berg-Demler-2010,Kitagawa-Broome-Fedrizzi-Rudner-Berg-Kassal-Aspuru-Demler-White-2012,Kitagawa_2012,Fuda-Funakawa-Suzuki-2017,Fuda-Funakawa-Suzuki-2018,Fuda-Funakawa-Suzuki-2019,Suzuki-Tanaka-2019,Matsuzawa-2020,Tanaka-2020}. In particular, if we set $p = 0,$ then this model becomes the usual \textbi{one-dimensional quantum walk model} considered in \cite{Ambainis-Bach-Nayak-Vishwanath-Watrous-2001,Konno-2002,Suzuki-2016}. 

\item If $m = 2,$ then $U_2$ turns out to be equivalent to the non-unitary evolution operator $U_{\textrm{mko}}$ given by \cref{equation: definition of evolution operator of MKO} in sense of the following theorem.
\end{itemize}

\begin{mtheorem}
\label{theorem: MKO model}
Let $U_{\textnormal{mko}}$ be given by \cref{equation: definition of evolution operator of MKO}, where we assume that four convergent $\R$-valued sequences $\gamma, \phi, \theta_1, \theta_2$ admit the two-sided limits of the form \cref{equation: existence of limits}. Then there exists a unitary self-adjoint operator $\varGamma_{\textnormal{mko}}$ on $\ell^2(\Z, \C^2),$ such that $(\varGamma_{\textnormal{mko}}, U_{\textnormal{mko}})$ forms a chiral pair. Moreover, the chiral pair $(\varGamma_{\textnormal{mko}}, U_{\textnormal{mko}})$ is unitarily equivalent to the chiral pair $(\varGamma_2, U_2),$ where the sequences $p,q,a,b$ are defined respectively by 
\begin{equation}
\tag{B1} 
\label{equation: mko substitution}
p := - \sin \theta_1(\cdot + 1), \quad q := - i\cos \theta_1(\cdot + 1), \quad a :=  \sin \theta_2, \quad b := i \cos \theta_2 e^{i (\phi + \phi(\cdot + 1))}.
\end{equation}
\end{mtheorem}

Complete classification of the two topological invariants $\ind(\varGamma_m, U_m)$ and $\ess(U_m)$ can be collectively found in the following theorem;
\begin{mtheorem}
\label{maintheorem: generalised mko}
If $(\varGamma_m, U_m)$ is the chiral pair in \cref{definition: Um}, then we have the following two assertions:
\begin{enumerate}
\item \textnormal{\textbf{Classification of the Witten index.}} For each $\star = \pm \infty,$ we let
\begin{equation}
\tag{C1}
\label{equation: definition of pgamma}
p_{\gamma}(\star) := \frac{p(\star)}{\sqrt{p(\star)^2 + |q(\star)|^2\cosh^2(2\gamma(\star))}}.
\end{equation}
Then the chiral pair $(\varGamma_m, U_{m})$ is Fredholm if and only if $|p_{\gamma}(\star)| \neq |a(\star)|$ for each $\star =  \pm \infty.$ In this case, we have the following index formula;
\begin{equation}
\tag{C2}
\label{equation: Witten index formula}
\frac{\ind (\varGamma_m, U_{m})}{m} \\
= 
\begin{cases}
0, & |p_{\gamma}(-\infty)| < |a(-\infty)|, \, |p_{\gamma}(+\infty)| < |a(+\infty)|, \\
\sgn p(+\infty), &  |p_{\gamma}(-\infty)| < |a(-\infty)|, \, |p_{\gamma}(+\infty)| > |a(+\infty)|, \\
- \sgn p(-\infty), & |p_{\gamma}(-\infty)| > |a(-\infty)|, \, |p_{\gamma}(+\infty)| < |a(+\infty)| , \\
\sgn p(+\infty) - \sgn p(-\infty), & |p_{\gamma}(-\infty)| > |a(-\infty)|, \,  |p_{\gamma}(+\infty)| > |a(+\infty)|,
\end{cases}
\end{equation}
where the sign function $\sgn : \R \to \{-1,1\}$ is defined by
\begin{equation}
\tag{C3}
\label{equation: definition of sign function}
\sgn x :=
\begin{cases}
\dfrac{x}{|x|}, & x \neq 0, \\
1, & x = 0.
\end{cases}
\end{equation}

\item \textnormal{\textbf{Classification of the essential spectrum.}}
For each $\star = \pm \infty,$ we let
\begin{align}
\label{equation: definition of s star}
\tag{C4}
s(\star) &:= \sgn(p(\star)a(\star)), \\
\label{equation: definition of Lambdapm}
\tag{C5}
\Lambda_\pm(\star) &:= |p(\star) a(\star)| \cosh(2\gamma(\star))  \pm |q(\star) b(\star)|, \\
\label{equation: definition of sigma star}
\tag{C6}
\sigma(\star) &:= \bigcup_{n \in \{ -1, +1\}} \left\{ \left( x + \sqrt{x^2 - 1} \right)^n \mathrel{} \middle| \mathrel{} s(\star)x \in [\Lambda_-(\star), \Lambda_+(\star)] \right\}.
\end{align}
Then the essential spectrum of $U_{m}$ can be written as $\ess(U_{m}) = \sigma(-\infty) \cup \sigma(+\infty).$ Furthermore, for each $\star = \pm \infty$ there exists a well-defined closed interval $[\gamma_-(\star), \gamma_+(\star)] \subseteq [0, \infty],$ such that the set $\sigma(\star)$ admits the following further classification:
\begin{itemize}
\item \textnormal{\textbf{Case I. }} If $|\gamma(\star)| \leq \gamma_-(\star),$ then $[\Lambda_-(\star),\Lambda_+(\star)] \subseteq  [-1,1],$ and so $\sigma(\star) \subseteq \T.$
\item \textnormal{\textbf{Case II. }} If $\gamma_-(\star) < |\gamma(\star)| < \gamma_+(\star),$ then $[\Lambda_-(\star), 1] \subseteq [-1,1]$ and $[1, \Lambda_+(\star)] \subseteq [1,\infty],$ and so $\sigma(\star)$ is a connected subset of $\T\cup \R$ containing $s(\star).$
\item \textnormal{\textbf{Case III. }} If $\gamma_+(\star) \leq |\gamma(\star)|,$ then $[\Lambda_-(\star),\Lambda_+(\star)] \subseteq  [1,\infty),$ and so $\sigma(\star) \subseteq \R.$
\end{itemize}
More explicitly, for each $\star = \pm \infty$ the closed interval $[\gamma_-(\star), \gamma_+(\star)]$ is given by
\begin{equation}
\tag{C7}
\label{equation: definition of gammaj}
\gamma_\pm(\star) 
:= \frac{1}{2} \cosh^{-1} \left( \frac{1 \pm |q(\star)b(\star)|}{|p(\star)a(\star)|}\right),
\end{equation}
where $\cosh^{-1}$ denotes the inverse function of $[0,\infty] \ni x \longmapsto \cosh x \in [1,\infty]$ with $1/0 := \infty$ by convention. 
\end{enumerate}
\end{mtheorem}

Explicit formulas for $\sigma(\star) \subseteq \T \cup \R$ in Cases I, II, III of \cref{maintheorem: generalised mko}~(ii) will be given shortly in \cref{section: gapless time-evolution}. This will allow us to classify $\sigma(\star)$ into the $6$ different cases as in \cref{figure: three cases}. 

\begin{remark}
If $\gamma$ is identically zero and if $m = 1,$ then $U_1$ is the \textit{unitary} time-evolution of a split-step quantum walk, and the formula for $\ind(\varGamma_1, U_1)$ can be found in \cite{Suzuki-Tanaka-2019,Matsuzawa-2020,Tanaka-2020}.  Similarly, under the same assumption, \cref{maintheorem: generalised mko}~(ii) coincides with \cite[Theorem 30]{Suzuki-Tanaka-2019} or \cite[Theorem B~(ii)]{Tanaka-2020}.
\end{remark}

\subsection{Discussion}
\label{section: gapless time-evolution}

Let us start with the following lemma;
\begin{lemma}
\label{lemma: essential spectrum of essentially unitary U}
Let $(\varGamma, U)$ be an abstract chiral pair on a Hilbert space $\cH,$ and let $Q$ be the imaginary part of $U.$ If $U$ is essentially unitary (i.e. $U^*U - 1, UU^* - 1$ are compact), then 
\begin{equation}
\label{equation: essential spectral mapping theorem}
\ess(Q) = \left\{\frac{z - z^*}{2i} \mathrel{} \middle| \mathrel{} z \in \ess(U)\right\}.
\end{equation}
That is, if $U$ is essentially unitary, then the chiral pair $(\varGamma, U)$ is Fredholm if and only if $U$ is essentially gapped in the sense of \cref{section: preliminaries}.
\end{lemma}
\begin{proof}
The formula \cref{equation: essential spectral mapping theorem} can be easily proved by using the spectral mapping theorem and the trigonometric polynomial $p(z) := (z - z^*)/(2i).$ We omit the proof, since an analogous argument can be found in \cite[Lemma 3.6 ]{Tanaka-2020}.  
\end{proof}

As in the following theorem, given an abstract chiral pair $(\varGamma,U),$ where $U$ may not necessarily be essentially unitary, the essential gappedness of $U$ is no longer an indispensable assumption to ensure the Fredholmness of the chiral pair $(\varGamma, U);$

\begin{theorem}
\label{theorem: main}
With the notation introduced in \cref{maintheorem: generalised mko}, suppose that the following hold true for each $\star = \pm \infty:$
\begin{equation}
\label{equation: main}
|p_\gamma(\star)| \neq |a(\star)|, \qquad \gamma_-(\star) < |\gamma(\star)| < \gamma_+(\star).
\end{equation}
Then $(\varGamma_m, U_m)$ is Fredholm, yet $U_m$ fails to be essentially gapped.
\end{theorem}
\begin{proof}
Let us first start with further classification of $\ess(U)$ given by \cref{maintheorem: generalised mko}~(ii). We consider the following $\R$-valued function $g$ defined on $(-\infty,-1] \cup [1, \infty);$
\[
g(x) := x + \sqrt{x^2 - 1}, \qquad x \in (-\infty,-1] \cup [1, \infty).
\]
\cref{figure: graph of g} shows the graphs of $g, 1/g;$
\begin{figure}[H]
\centering
\label{graph: graph of g}
\begin{tikzpicture}[baseline=(current  bounding  box.center)]
\begin{axis}[legend pos = north west,  ytick= {-1,0,1}, yticklabel style={anchor = north west}, xtick= {-1,0,1}, xticklabel style={anchor = north west}, xmin= -3, xmax=3, ymin=-3, ymax=3, legend pos = north west, axis lines=center, xlabel=$x$, ylabel=$y$, xlabel style={anchor = west}
]
\addlegendentry{$y = g(x)$}
	\addplot[color = black, samples=200,domain=1:4, very thick]{x + sqrt(x^2 - 1)};
	\addplot[color = black, samples=200,domain=-1:-4, very thick, forget plot]{x + sqrt(x^2 - 1)};
	\addplot [color = black!20!white, domain=-4:4,samples=50, smooth, dashed, forget plot]({x},{1}); 
	\addplot [color = black!20!white, domain=-3:3,samples=50, smooth, dashed, forget plot]({1},{x}); 
	\addplot [color = black!20!white, domain=-4:4,samples=50, smooth, dashed, forget plot]({x},{-1}); 
	\addplot [color = black!20!white, domain=-3:3,samples=50, smooth, dashed, forget plot]({-1},{x}); 
\addlegendentry{$y = \frac{1}{g(x)}$}
	\addplot[color = lightgray, densely dashed, samples=200,domain=1:4,  very thick]{1/(x + sqrt(x^2 - 1))};
	\addplot[color = lightgray, densely dashed,samples=200,domain=-1:-4, very thick, forget plot]{1/(x + sqrt(x^2 - 1))};
\end{axis}
\end{tikzpicture}
\caption{The black graph corresponds to $g,$ and the gray graph corresponds to $g^{-1}.$}
\label{figure: graph of g}
\end{figure}
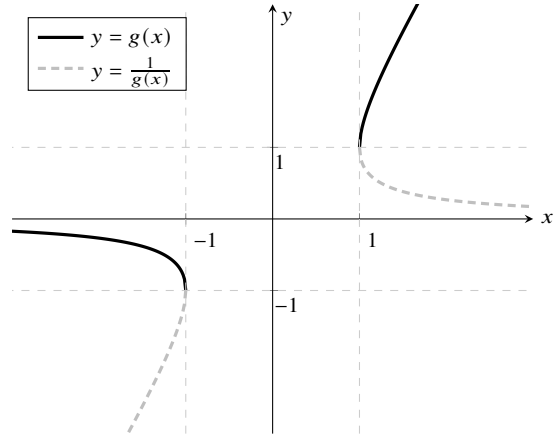
Evidently, $g(x)g(-x)^{-1} = -1$ for $|x| \geq 1.$ It follows that for each $\star = \pm \infty$ the set $\sigma(\star)$ admits the following further classification;
\begin{itemize}
\item \textnormal{\textbf{Case I. }} If $|\gamma(\star)| \leq \gamma_-(\star),$ then
\[
\sigma(\star) 
=
\begin{cases}
\{z \in \T \mid \Re z \in [\Lambda_-(\star), \Lambda_+(\star)]\},    &  s(\star) = 1, \\
\{z \in \T \mid \Re z \in [-\Lambda_+(\star), -\Lambda_-(\star)]\}, & s(\star) = -1.
\end{cases}
\]

\item \textnormal{\textbf{Case II. }} If $\gamma_-(\star) < |\gamma(\star)| < \gamma_+(\star),$ then
\[
\sigma(\star) 
=
\begin{cases}
\{z \in \T \mid \Re z \in [\Lambda_-(\star), 1]\} \cup [g(\Lambda_+(\star))^{-1}, g(\Lambda_+(\star))],& s(\star) = 1, \\
\{z \in \T \mid \Re z \in [-1, -\Lambda_-(\star)]\} \cup [-g(\Lambda_+(\star)), -g(\Lambda_+(\star))^{-1}],& s(\star) = -1.
\end{cases}
\]

\item \textnormal{\textbf{Case III. }} If $\gamma_+(\star) \leq |\gamma(\star)|,$ then
\[
\sigma(\star) 
=
\begin{cases}
[g(\Lambda_+(\star))^{-1},g(\Lambda_-(\star))^{-1}] \cup [g(\Lambda_-(\star)), g(\Lambda_+(\star))], & s(\star) = 1, \\
[-g(\Lambda_+(\star)),-g(\Lambda_-(\star))] \cup [-g(\Lambda_-(\star))^{-1},-g(\Lambda_+(\star))^{-1}],& s(\star) = -1.
\end{cases}
\]
\end{itemize}
That is, $\sigma(\star)$ is classified into the $6$ different cases as in \cref{figure: three cases}. It immediately follows from \cref{maintheorem: generalised mko} and \cref{equation: main} that $(\varGamma_m, U_m)$ is Fredholm, and that $\ess(U_m) = \sigma(-\infty) \cup \sigma(+\infty).$ In particular, for each $\star = \pm \infty$ the set $\sigma(\star)$ is classified as Case II. That is, each $\sigma(\star)$ is a connected subset of $\T \cup \R$ containing either $-1$ or $+1,$ and so $U_m$ fails to be essentially gapped. The claim follows.
\end{proof}

The current section concludes with the following two numerical examples:

\begin{example}
\label{example: first example}
Let $(\varGamma_m, U_m)$ be the chiral pair in \cref{maintheorem: generalised mko}. Let 
\[
p_0 := 0.2, \qquad a_0 := 0.1, \qquad \gamma_0 := 0.4. 
\]
If $a(\pm \infty) := \pm a_0$ and $p(\pm \infty) := \pm p_0,$ then \cref{equation: definition of gammaj} becomes
\begin{align*}
\gamma_-(-\infty) 
&= \gamma_-(+ \infty) 
= \frac{1}{2} \cosh^{-1} \left(\frac{1  - \sqrt{1 - p_0^2}\sqrt{1 - a_0^2}}{|p_0 a_0|}\right)
= 0.350396, \\
\gamma_+(-\infty) 
&= \gamma_+(+ \infty)  
= \frac{1}{2} \cosh^{-1} \left(\frac{1 + \sqrt{1 - p_0^2}\sqrt{1 - a_0^2}}{|p_0 a_0|}\right)
= 2.64283.
\end{align*}
If we let $\gamma(\pm \infty) := \gamma_0,$ then $\gamma_-(\pm \infty) < |\gamma_0| < \gamma_+(\pm \infty).$ It follows from \cref{equation: definition of sigma star} that $\ess(U_m) = \sigma(-\infty) = \sigma(+\infty),$ since $s(-\infty) = s(+\infty) = 1.$ More precisely, the set $\ess(U_m) = \sigma(\pm \infty)$ is classified as Case II:
\begin{align*}
\Lambda_\pm &:= |p_0 a_0| \cosh(2\gamma_0)  \pm \sqrt{1 - p_0^2}\sqrt{1 - a_0^2}, \\
\ess(U_m) &:= \{z \in \T \mid \Re z \in [\Lambda_-, 1]\} \cup [g(\Lambda_+)^{-1}, g(\Lambda_+)].
\end{align*}
The black region in Cases II of \cref{figure: three cases} depicts the connected subset $\ess(U_m)$ of $\T \cup \R$ containing $1.$ It follows that $U_m$ is not essentially gapped. Furthermore, \cref{equation: definition of pgamma} becomes
\[
|p_\gamma (\pm \infty)| = 
\frac{|p_0|}{\sqrt{p_0^2 + (1 - p_0) \cosh^2(2\gamma_0)}} = 0.150876 > |a(\pm \infty)| = 0.1.
\]
It follows that $(\varGamma_m, U_m)$ is Fredholm, and $\ind(\varGamma_m, U_m) = m(+1 - (-1)) = 2m$ by the index formula \cref{equation: Witten index formula}. That is, we have constructed the Fredholm chiral pair $(\varGamma_m, U_m),$ in such a way that $U_m$ fails to be essentially gapped, yet $\ind(\varGamma_m, U_m) = 2m$ is well-defined.
\end{example}

\begin{example}
\label{example: second example}
Let $U_{\textnormal{mko}}$ be the non-unitary evolution operator given by \cref{equation: definition of evolution operator of MKO}, where we assume the existence of the two-sided limits \cref{equation: existence of limits}. We define $p,q,a,b$ according to \cref{equation: mko substitution}, and \cref{theorem: MKO model} asserts $(\varGamma_{\textnormal{mko}},U_{\textnormal{mko}}) \simeq (\varGamma_2, U_2),$ where $\simeq$ denotes unitary equivalence of chiral pairs. As in \cref{example: first example}, we choose $\theta_1(\pm \infty), \theta_2(\pm \infty), \gamma(\pm \infty)$ in such a way that 
\[
p(\pm \infty) = \pm 0.2, \qquad a(\pm \infty) = \pm 0.1, \qquad \gamma(\pm \infty) = 0.4. 
\]
It follows that $(\varGamma_{\textnormal{mko}},U_{\textnormal{mko}}) \simeq (\varGamma_2, U_2)$ is Fredholm, and $\ind(\varGamma_{\textnormal{mko}},U_{\textnormal{mko}}) = 4.$ Furthermore, $U_{\textnormal{mko}}$ is essentially gapless.
\end{example}

\section{Proofs of the main theorems}
\label{section: proof of main theorem}

\subsection{Unitary invariance of the Witten index (\texorpdfstring{\cref{lemma: unitary invariance of the witten index}}{Lemma})}

We prove the unitary invariance of the Witten index (\cref{lemma: unitary invariance of the witten index}). Note that the special case of this invariance principle for unitary $U$ can be found in \cite[Corollary 3.6]{Suzuki-2019}, the proof of which makes use of a spectral mapping theorem for chirally symmetric unitary operators
\cite{Segawa-Suzuki-2016,Segawa-Suzuki-2019}. We give the following direct proof instead;

\begin{proof}[Proof of \cref{lemma: unitary invariance of the witten index}]
Let $(\varGamma,U), (\varGamma',U')$ be two unitarily equivalent chiral pairs on a Hilbert space $\cH.$  That is, there exists a unitary operator $\epsilon$ on $\cH,$ such that
$
(\varGamma',U') = (\epsilon^* \varGamma \epsilon, \epsilon^* U \epsilon).
$
Let $\cH_\pm := \ker(\varGamma \mp 1),$ and let $\cH'_\pm := \ker(\varGamma' \mp 1).$ We may assume  that the operator $\epsilon$ admits the following block-operator matrix representation;
\begin{align*}
&\epsilon = 
\begin{pmatrix}
\epsilon_{+} &  \epsilon_{-+} \\
\epsilon_{+-} & \epsilon_{-}
\end{pmatrix}, 
&&
\begin{aligned}
\epsilon_{+}&: \cH'_+ \to \cH_+,  \qquad &  \epsilon_{-+}&: \cH'_- \to \cH_+, \\
\epsilon_{+-}&: \cH'_+ \to \cH_-, \qquad &  \epsilon_{-}&: \cH'_- \to \cH_-.
\end{aligned}
\end{align*}
Recall that the operators $U,U'$ admit the following block-operator matrix representations respectively according to \cref{section: preliminaries}:
\[
U = 
\begin{pmatrix}
R_1 & iQ_0^* \\
iQ_0 & R_2
\end{pmatrix}_{\cH_+ \oplus \cH_-}, \qquad 
U'= 
\begin{pmatrix}
R'_1 & i(Q'_0)^* \\
iQ'_0 & R'_2
\end{pmatrix}_{\cH'_+ \oplus \cH'_-}.
\]
Since $0 = \epsilon \varGamma' - \varGamma \epsilon,$ where $\varGamma = 1 \oplus (-1)$ and $\varGamma' = 1 \oplus (-1),$ we obtain
\[
0 =
\begin{pmatrix}
\epsilon_{+} &  \epsilon_{-+} \\
\epsilon_{+-} & \epsilon_{-}
\end{pmatrix}
\begin{pmatrix}
1 & 0 \\
0 & -1
\end{pmatrix}
-
\begin{pmatrix}
1 & 0 \\
0 & -1
\end{pmatrix}
\begin{pmatrix}
\epsilon_{+} &  \epsilon_{-+} \\
\epsilon_{+-} & \epsilon_{-}
\end{pmatrix}
=
\begin{pmatrix}
\epsilon_{+} &  -\epsilon_{-+} \\
\epsilon_{+-} & -\epsilon_{-}
\end{pmatrix}
+
\begin{pmatrix}
-\epsilon_{+} &  -\epsilon_{-+} \\
\epsilon_{+-} & \epsilon_{-}
\end{pmatrix}
=
\begin{pmatrix}
0 &  -2\epsilon_{-+} \\
2\epsilon_{+-} & 0
\end{pmatrix}.
\]
This implies $\epsilon = \epsilon_{+} \oplus \epsilon_{-} : \cH_+' \oplus \cH'_- \to \cH_+ \oplus \cH_-,$ and so 
\[
\epsilon^* U \epsilon = (\epsilon^*_{11} \oplus \epsilon_{-}^*) U (\epsilon_{+} \oplus \epsilon_{-}) =
\begin{pmatrix}
\epsilon^*_{+}  R_1 \epsilon_{+} & i \epsilon^*_{+}  Q_0^* \epsilon_{-} \\
i \epsilon^*_{-}  Q_0 \epsilon_{+} & \epsilon^*_{-} R_2 \epsilon_{-}\\
\end{pmatrix} =
\begin{pmatrix}
R'_1 & i (Q'_{0})^* \\
i Q'_{0} & R'_{2} \\
\end{pmatrix}.
\]
Since $Q'_{0} = \epsilon^*_{-}  Q_0 \epsilon_{+},$ where $\epsilon_{+},\epsilon_{-}$ are unitary, we have that $Q_{0}$ is Fredholm if and only if so is $Q'_0.$ In this case, $\ind Q_{0} = \ind Q'_0.$ The claim follows.
\end{proof}

\subsection{Unitary transform of the Mochizuki-Kim-Obuse model (\texorpdfstring{\cref{theorem: MKO model}}{Theorem B})}
\label{section: MKO model}

\begin{proof}[Proof of \cref{theorem: MKO model}]
We shall make use of the fact that the operator $S$ commutes with any diagonal block-operator matrices in this proof. We have
\begin{align*}
U_{\textnormal{mko}} C_1^{-1}
&= S (G \Phi) C_2 S (G^{-1} \Phi)  \\
&= S
\begin{pmatrix}
e^{\gamma + i \phi} & 0  \\
0 & e^{-\gamma(\cdot + 1)-i\phi(\cdot + 1)}
\end{pmatrix}
\begin{pmatrix}
\cos \theta_2 & i \sin \theta_2 \\
i \sin \theta_2 & \cos \theta_2 
\end{pmatrix}
S \begin{pmatrix}
e^{-\gamma + i \phi} & 0  \\
0 & e^{\gamma(\cdot + 1) -i\phi(\cdot + 1)}
\end{pmatrix}\\
&=
S
\begin{pmatrix}
e^{\gamma + i \phi} & 0  \\
0 & e^{-\gamma(\cdot + 1)-i\phi(\cdot + 1)}
\end{pmatrix}
\begin{pmatrix}
\cos \theta_2 & i \sin \theta_2 \\
i \sin \theta_2 & \cos \theta_2 
\end{pmatrix}
\begin{pmatrix}
e^{-\gamma(\cdot + 1) + i \phi(\cdot + 1)} & 0  \\
0 & e^{\gamma -i\phi}
\end{pmatrix}S \\
&=
S
\begin{pmatrix}
\cos \theta_2 e^{\gamma  -\gamma(\cdot + 1) +i (\phi + \phi(\cdot + 1))} & i \sin \theta_2 e^{2\gamma}\\
i \sin \theta_2e^{-2\gamma(\cdot + 1)}  & \cos \theta_2 e^{\gamma -\gamma(\cdot + 1)-i(\phi +\phi(\cdot + 1))}
\end{pmatrix}
S,
\end{align*}
where the third equality follows from $L^{\pm 1} \Psi = \Psi(\cdot \pm 1)$ for any $\Psi \in \ell^2(\Z).$ If $\sigma_2 = 
\begin{pmatrix}
0 & -i  \\
i & 0
\end{pmatrix}$ 
denotes the second Pauli matrix, then $\sigma_2^2 = 1,$ and so
\begin{align*}
U_{\textnormal{mko}}  &= 
(S\sigma_2)
\sigma_2 
\begin{pmatrix}
\cos \theta_2 e^{\gamma  -\gamma(\cdot + 1) +i (\phi + \phi(\cdot + 1))} & i \sin \theta_2 e^{2\gamma}\\
i \sin \theta_2e^{-2\gamma(\cdot + 1)}  & \cos \theta_2 e^{\gamma -\gamma(\cdot + 1)-i(\phi +\phi(\cdot + 1))}
\end{pmatrix}
(S\sigma_2) (\sigma_2 C_1) \\
&=
(S\sigma_2)
\begin{pmatrix}
\sin \theta_2e^{-2\gamma(\cdot + 1)}  & -i\cos \theta_2 e^{\gamma -\gamma(\cdot + 1)-i(\phi +\phi(\cdot + 1))} \\
i\cos \theta_2 e^{\gamma  -\gamma(\cdot + 1) +i (\phi + \phi(\cdot + 1))} & - \sin \theta_2 e^{2\gamma}\\
\end{pmatrix}
(S\sigma_2) (\sigma_2 C_1) \\
&=
(S\sigma_2)
\begin{pmatrix}
a e^{-2\gamma(\cdot + 1)}  & b^* e^{\gamma -\gamma(\cdot + 1)} \\
be^{\gamma  -\gamma(\cdot + 1)} & - a e^{2\gamma}\\
\end{pmatrix}
(S\sigma_2) (\sigma_2 C_1).
\end{align*}
If we let $\eta := (\sigma_2 C_1) (S \sigma_2),$ where
$\sigma_2 C_1$ and $S\sigma_2$ are unitary involutions, then 
\begin{align*}
\eta^* U_{\textnormal{mko}} \eta 
&= 
(S \sigma_2)
(\sigma_2 C_1) 
(S\sigma_2)
\begin{pmatrix}
a e^{-2\gamma(\cdot + 1)}  & b^* e^{\gamma -\gamma(\cdot + 1)} \\
be^{\gamma  -\gamma(\cdot + 1)} & - a e^{2\gamma}\\
\end{pmatrix}
(S\sigma_2) (\sigma_2 C_1)(\sigma_2 C_1) (S \sigma_2) \\
&= 
(S \sigma_2)
(\sigma_2 C_1) 
(S\sigma_2)
\begin{pmatrix}
a e^{-2\gamma(\cdot + 1)}  & b^* e^{\gamma -\gamma(\cdot + 1)} \\
be^{\gamma  -\gamma(\cdot + 1)} & - a e^{2\gamma}\\
\end{pmatrix}.
\end{align*}
It remains to compute $(S \sigma_2)(\sigma_2 C_1)(S\sigma_2);$ 
\[
(S \sigma_2) (\sigma_2 C_1) (S\sigma_2)
= 
\begin{pmatrix}
0 & -iL \\
iL^{-1} & 0  \\
\end{pmatrix}
\begin{pmatrix}
\sin \theta_1 & -i\cos \theta_1  \\
i\cos \theta_1 & -\sin \theta_1  \\
\end{pmatrix}
\begin{pmatrix}
0 & -iL \\
iL^{-1} & 0  \\
\end{pmatrix}
= \varGamma_2.
\]
If we let $\varGamma_{\textrm{mko}} := \eta \varGamma_2 \eta^*,$ then $\eta^* \varGamma_{\textrm{mko}} \eta  = \varGamma_2.$ The claim follows.
\end{proof}

\subsection{Classification of the topological invariants (\texorpdfstring{\cref{maintheorem: generalised mko}}{Theorem })}

\subsubsection{Strictly local operators}

To prove \cref{maintheorem: generalised mko}, let us first introduce one preliminary concept beforehand. With the obvious orthogonal decomposition $\ell^2(\Z, \C^n) = \bigoplus_{j=1}^n \ell^2(\Z, \C)$ in mind, we shall consider an operator of the form
\begin{align}
\label{equation2: characterisation of strict locality} 
A 
= 
\begin{pmatrix}
\sum^k_{y=-k} a_{11}(y, \cdot) L^{y} & \dots & \sum^k_{y=-k} a_{1n}(y, \cdot) L^{y} \\
\vdots & \ddots& \vdots \\
\sum^k_{y=-k} a_{n1}(y, \cdot) L^{y} & \dots & \sum^k_{y=-k} a_{nn}(y, \cdot) L^{y} \\
\end{pmatrix},
\end{align}
where $k$ is a finite natural number, and where each $a_{ij}(y, \cdot) = (a_{ij}(y, x))_{x \in \Z}$ is an arbitrary bounded $\C$-valued sequence viewed as a multiplication operator on $\ell^2(\Z, \C) = \bigoplus_{x \in \Z} \C.$ An operator the form \cref{equation2: characterisation of strict locality} will be referred to as a (one-dimensional) \textbi{strictly local operator} following \cite[\textsection 1.2]{Cedzich-Geib-Stahl-Velazquez-Werner-Werner-2018}.

\begin{theorem}[{\cite[Theorem A]{Tanaka-2020}}]
\label{theorem: topological invariants of strictly local operators}
Let $A$ be a strictly local operator of the form \cref{equation2: characterisation of strict locality} with the property that the following two-sided limits exist:
\begin{equation}
\label{equation: two-phase assumptions}
a_{ij}(y, \pm \infty) := \lim_{x \to \pm \infty} a_{ij}(y,x) \in \C, \qquad i,j = 1, \dots, n ,\ -k \leq y \leq k.
\end{equation}
Let
\begin{align}
\label{equation: definition of Apm}
A(\pm \infty) 
&:=  
\begin{pmatrix}
\sum^k_{y=-k} a_{11}(y, \pm \infty) L^{y} & \dots & \sum^k_{y=-k} a_{1n}(y, \pm \infty) L^{y} \\
\vdots & \ddots& \vdots \\
\sum^k_{y=-k} a_{n1}(y, \pm \infty) L^{y} & \dots & \sum^k_{y=-k} a_{nn}(y, \pm \infty) L^{y} \\
\end{pmatrix}, \\
\label{equation: definition of Fourier transform of A}
\hat{A}(z,\pm \infty)
&:=  
\begin{pmatrix}
\sum^k_{y=-k} a_{11}(y, \pm \infty) z^{y} & \dots & \sum^k_{y=-k} a_{1n}(y, \pm \infty) z^{y} \\
\vdots & \ddots& \vdots \\
\sum^k_{y=-k} a_{n1}(y, \pm \infty) z^{y} & \dots & \sum^k_{y=-k} a_{nn}(y, \pm \infty) z^{y} \\
\end{pmatrix}, \qquad z \in \T.
\end{align}
Then the following assertions hold true:
\begin{enumerate}
\item We have that $A$ is Fredholm if and only if $\T \ni z \longmapsto \det \hat{A}(z,\star) \in \C$ is nowhere vanishing on $\T$ for each $\star = \pm \infty.$ In this case, the Fredholm index of $A$ is given by
\begin{equation}
\label{equation: bulk-edge correspondence}
\ind(A) = \wn \left(\det \hat{A}(\cdot,+\infty) \right) - \wn \left(\det \hat{A}(\cdot, -\infty) \right),
\end{equation}
where $\wn \left(\det \hat{A}(\cdot, \star) \right)$ denotes the winding number of the continuous function $\T \ni z \longmapsto \det \hat{A}(z,\star) \in \C$ with respect to the origin.

\item The essential spectrum of $A$ is given by
\begin{align*}
&\ess(A) = \ess(A(- \infty)) \cup \ess(A(+ \infty)), \\
&\ess(A(\star)) = \bigcup_{z \in \T} \sigma(\hat{A}(z,\star)), \qquad \star = \pm \infty.
\end{align*}
\end{enumerate}
\end{theorem}

\cref{theorem: topological invariants of strictly local operators} can be viewed as an abstract form of the 
one-dimensional bulk-boundary correspondence \cite[Corollary 4.3]{Cedzich-Geib-Grunbaum-Stahl-Velazquez-Werner-Werner-2018} (see \cite[\textsection 2]{Tanaka-2020} for details).

\subsubsection{Proof of \texorpdfstring{\cref{maintheorem: generalised mko}~(i)}{Theorem C (i)}}
\label{section: proof of the index formula}

\begin{notation}
We shall make use of the notation introduce in \cref{definition: Um}. For notational simplicity, we use the following notation throughout \cref{section: proof of the index formula};
\[
(\varGamma,U) := (\varGamma_m,U_m), \qquad 
C := 
\begin{pmatrix}
\alpha_1 & \beta^* \\
\beta & \alpha_2
\end{pmatrix} :=
\begin{pmatrix}
e^{-2 \gamma(\cdot + 1)}a & e^{\gamma - \gamma(\cdot + 1)}b^* \\
e^{\gamma - \gamma(\cdot + 1)}b & -e^{2 \gamma}a
\end{pmatrix}.
\]
With the above notation, the operator $U$ can be written as $U = \varGamma C.$
\end{notation}

In order to compute $\ind(\varGamma, U)$ we shall closely follow \cite[\textsection 3.2]{Tanaka-2020}. Note first that the underlying Hilbert space $\ell^2(\Z,\C^2)$ admits the following two orthogonal decompositions: 
\[
\ell^2(\Z,\C^2) 
= \ker(\varGamma - 1) \oplus  \ker(\varGamma + 1)
= \ell^2(\Z) \oplus \ell^2(\Z),
\]
where $\ker(\varGamma \mp 1) \neq \ell^2(\Z).$ On one hand, the imaginary part $Q$ of $U$ admits an off-diagonal block operator matrix representation with respect to the former decomposition as in the second equality of \cref{equation: representation of R and Q}, where the Fredholm index of $Q_0 : \ker(\varGamma - 1) \to \ker(\varGamma + 1)$ is by definition $\ind(\varGamma, U).$ On the other hand, the same operator $Q$ can \textit{not} be expressed as an off-diagonal block-operator matrix with respect to the latter decomposition. The unitary invariance of the Witten index (\cref{lemma: unitary invariance of the witten index}) motivates us to construct a unitary operator $\epsilon : \ell^2(\Z) \to \ell^2(\Z),$ in such a way that the imaginary part $\epsilon^* Q \epsilon$ of the new chiral pair $(\epsilon^* \varGamma \epsilon, \epsilon^* U \epsilon)$ become off-diagonal with respect to $\ell^2(\Z) \oplus \ell^2(\Z).$

\begin{lemma} 
\label{lemma: wada decomposition}
Let $R,Q$ be the real and imaginary parts of $U$ respectively. For each $x \in \Z,$ let $\theta(x)$ be any real number satisfying $q(x) = |q(x)|e^{i \theta(x)},$ and let $p_\pm(x) := \sqrt{1 \pm p(x)}.$ Let
\begin{align}
\label{equation1: definition of Qepsilon}
-2i Q_{\epsilon_0} &:= p_+ e^{i \theta}L^{m} \beta p_+ - p_- \beta^*  L^{-m} e^{-i \theta}p_- - |q|(\alpha_1 - \alpha_2(\cdot + m)), \\
\label{equation1: definition of Repsilon}
2R_{\epsilon_1} &:= p_- e^{i \theta}L^{m} \beta p_+ + p_+ \beta^* L^{-m} e^{-i \theta} p_- + p_+^2 \alpha_1 + p_-^2\alpha_2(\cdot + m), \\
\label{equation2: definition of Repsilon}
2R_{\epsilon_2} &:= p_+ e^{i \theta} L^{m} \beta p_- + p_- \beta^* L^{-m} e^{-i \theta} p_+ - p_-^2\alpha_1 - p_+^2 \alpha_2(\cdot + m).
\end{align}
Then there exists a unitary operator $\epsilon$ on $\ell^2(\Z,\C^2),$ such that the following block-operator matrix representations hold true with respect to $\ell^2(\Z,\C^2) = \ell^2(\Z) \oplus \ell^2(\Z):$
\begin{align*}
&\epsilon^* \varGamma  \epsilon 
= 
\begin{pmatrix}
1 & 0 \\
0 & -1
\end{pmatrix}, 
&&\epsilon^* U \epsilon 
= 
\begin{pmatrix}
R_{\epsilon_1} & iQ_{\epsilon_0}^* \\
iQ_{\epsilon_0} & R_{\epsilon_2}
\end{pmatrix}, 
&&\epsilon^* R \epsilon 
=  
\begin{pmatrix}
R_{\epsilon_1} & 0 \\
0 & R_{\epsilon_2}
\end{pmatrix}, 
&&\epsilon^* Q \epsilon 
=  
\begin{pmatrix}
0 & Q_{\epsilon_0}^* \\
Q_{\epsilon_0} & 0
\end{pmatrix},
\end{align*}
Moreover, the chiral pair $(\varGamma, U)$ is Fredholm if and only if $Q_{\epsilon_0}$ is Fredholm. In this case, 
\begin{equation}
\label{equation: first index formula}
\ind(\varGamma, U) = \ind Q_{\epsilon_0}.
\end{equation}
\end{lemma}
As we shall see below, the derivation of the index formula \cref{equation: first index formula} only requires the boundedness of the given sequences $\gamma, p, a, q, b,$ and so \cref{equation: limits of sequences} turns out to be redundant. Note, however, that this assumption \cref{equation: limits of sequences}  is necessary to prove the index formula \cref{equation: Witten index formula}.
\begin{proof}
Note first that $\varGamma$ can be written as
\[
\varGamma
=
\begin{pmatrix}
p & qL^{m} \\
L^{-m}q^*  & -p(\cdot - m)
\end{pmatrix}
= 
\begin{pmatrix}
1 & 0\\
0  & L^{-m}
\end{pmatrix}
\begin{pmatrix}
p   & q \\
q^* & -p
\end{pmatrix}
\begin{pmatrix}
1 & 0\\
0  & L^{m}
\end{pmatrix},
\]
where the middle matrix on the right hand side of the second equality admits the following diagonalisation. For each $x \in \Z$ we have
\begin{equation}
\label{equation: diagonalisation of unitary involutory matrices}
\epsilon_0(x)^*
\begin{pmatrix}
p(x) & q(x) \\
q(x)^*  & -p(x)
\end{pmatrix}
\epsilon_0(x)
=
\begin{pmatrix}
1 & 0 \\
0 & -1
\end{pmatrix}, \quad 
\epsilon_0(x)
:=
\frac{1}{\sqrt{2}}
\begin{pmatrix}
1 & 0 \\
0 & e^{-i \theta(x)}
\end{pmatrix}
\begin{pmatrix}
p_+(x) & -p_-(x) \\
p_-(x)   & p_+(x) 
\end{pmatrix}.
\end{equation}
Since $\epsilon_0 := \bigoplus_{x \in \Z} \epsilon_0(x)$ is unitary, the following operator is also unitary;
\[
\epsilon :=
\begin{pmatrix}
1 & 0 \\
0  & L^{-m}
\end{pmatrix}
\epsilon_0 
=
\frac{1}{\sqrt{2}}
\begin{pmatrix}
1 & 0 \\
0  & L^{-m}e^{-i \theta}
\end{pmatrix}
\begin{pmatrix}
p_+  & -p_- \\
p_-  & p_+
\end{pmatrix}.
\]
It follows from the first equality that
\[
\epsilon^* \varGamma  \epsilon 
= 
\epsilon_0^*
\begin{pmatrix}
1 & 0 \\
0  & L^{m}
\end{pmatrix}
\begin{pmatrix}
1 & 0 \\
0  & L^{-m}
\end{pmatrix}
\begin{pmatrix}
p   & q \\
q^* & -p
\end{pmatrix}
\begin{pmatrix}
1 & 0 \\
0  & L^{m}
\end{pmatrix}
\begin{pmatrix}
1 & 0 \\
0  & L^{-m}
\end{pmatrix}
\epsilon_0 \\
=
\epsilon_0^*
\begin{pmatrix}
p & q \\
q^*  & -p
\end{pmatrix}
\epsilon_0 = 
\begin{pmatrix}
1 & 0 \\
0 & -1
\end{pmatrix},
\]
where the last equality follows from \cref{equation: diagonalisation of unitary involutory matrices}.

Given a bounded operator $X$ on $\ell^2(\Z,\C^2),$ we introduce the shorthand $X_\epsilon := \epsilon^* X \epsilon.$ With this convention in mind, we have $[\varGamma _\epsilon, R_\epsilon] = 0 = \{\varGamma _\epsilon, Q_\epsilon\},$  where $\varGamma _\epsilon = 1 \oplus (-1)$ with respect to $\ell^2(\Z,\C^2) = \ell^2(\Z) \oplus  \ell^2(\Z).$ It follows that we have the following representations:
\begin{align}
\label{equation: epsilon representation}
R_\epsilon 
&=  
\begin{pmatrix}
R'_{\epsilon_1} & 0 \\
0 & R'_{\epsilon_2}
\end{pmatrix}, 
&
Q_\epsilon 
&=  
\begin{pmatrix}
0 & (Q'_{\epsilon_0})^* \\
Q'_{\epsilon_0} & 0
\end{pmatrix},
&
U_\epsilon 
&= R_\epsilon + iQ_\epsilon = 
\begin{pmatrix}
R'_{\epsilon_1} & i(Q'_{\epsilon_0})^* \\
iQ'_{\epsilon_0} & R'_{\epsilon_2}
\end{pmatrix}.
\end{align}
It remains to show that the three operators $Q'_{\epsilon_0},R'_{\epsilon_1},R'_{\epsilon_2}$ introduced above coincide with the ones defined by the formulas \crefrange{equation1: definition of Qepsilon}{equation2: definition of Repsilon}. Note that
\begin{align}
\label{equation1: Cepsilon}
2C_\epsilon 
&= 
\varGamma _\epsilon (2U_\epsilon)
=
\begin{pmatrix}
1 & 0 \\
0 & -1
\end{pmatrix}
\begin{pmatrix}
2R'_{\epsilon_1} & 2i(Q'_{\epsilon_0})^* \\
2iQ'_{\epsilon_0}   & 2R'_{\epsilon_2}  \\
\end{pmatrix}
=
\begin{pmatrix}
2R'_{\epsilon_1}     & 2i(Q'_{\epsilon_0})^* \\
-2iQ'_{\epsilon_0}   & -2R'_{\epsilon_2}  \\
\end{pmatrix}.
\end{align}
It remains to compute $2C_\epsilon.$ We have
\begin{align*}
2\epsilon^*
\begin{pmatrix}
\alpha_1 & 0 \\
0 & \alpha_2 \\
\end{pmatrix}
\epsilon
&= 
\begin{pmatrix}
p_+^2 \alpha_1 + p_-^2 \alpha_2(\cdot + m) & -|q|(\alpha_1 - \alpha_2(\cdot + m)) \\
-|q|(\alpha_1 - \alpha_2(\cdot + m))  & p_-^2 \alpha_1 + p_+^2\alpha_2(\cdot + m)
\end{pmatrix}, \\
2\epsilon^*
\begin{pmatrix}
0 & \beta^* \\
\beta & 0 \\
\end{pmatrix}
\epsilon
&= 
\begin{pmatrix}
p_- e^{i \theta} L^{m} \beta p_+ + p_+ \beta^* L^{-m}e^{-i \theta} p_-    & -p_- e^{i \theta}L^{m} \beta p_- + p_+ \beta^* L^{-m}e^{-i \theta} p_+   \\
p_+ e^{i \theta}L^{m} \beta p_+ - p_- \beta^* L^{-m}e^{-i \theta} p_-  &
-p_+ e^{i \theta}L^{m} \beta p_- - p_- \beta^* L^{-m} e^{-i \theta}p_+ 
\end{pmatrix}.
\end{align*}
It follows from the above two equalities that
\begin{equation}
\label{equation2: Cepsilon}
2C_\epsilon = 
2\epsilon^*
\begin{pmatrix}
\alpha_1 & 0 \\
0 & \alpha_2 \\
\end{pmatrix}
\epsilon +
2\epsilon^*
\begin{pmatrix}
0 & \beta^* \\
\beta & 0 \\
\end{pmatrix}
\epsilon
=
\begin{pmatrix}
2R_{\epsilon_1} & 2i Q_{\epsilon_0}^* \\
-2i Q_{\epsilon_0} & -2R_{\epsilon_2}
\end{pmatrix}
\end{equation}
By comparing \cref{equation1: Cepsilon} with \cref{equation2: Cepsilon}, we see that \cref{equation: epsilon representation} also holds true without the superscript $'$.

Note that $\ell^2(\Z,\C^2) = \ell^2(\Z) \oplus \ell^2(\Z)$ can be identified with the orthogonal sum $\ell^2(\Z) \oplus \{0\} \oplus \{0\} \oplus \ell^2(\Z)$ through the following unitary transform;
\[
\ell^2(\Z,\C^2) \ni (\Psi_1, \Psi_2) \longmapsto (\Psi_1, 0,0,\Psi_2) \in \ell^2(\Z) \oplus \{0\} \oplus \{0\} \oplus \ell^2(\Z).
\]
It is then easy to show that the operator $Q_\epsilon$ admits the following block-operator matrix representations:
\begin{equation}
\label{equation1: representation of Qepsilon}
Q_\epsilon
=
\begin{pmatrix}
0 & Q_{\epsilon_0}^* \\
Q_{\epsilon_0} & 0
\end{pmatrix}_{\ell^2(\Z) \oplus \ell^2(\Z)}
=
\begingroup
\setlength\arraycolsep{5pt}
\begin{pmatrix}
0 & 0 & 0 & Q_{\epsilon_0} \\
0 & 0 & \textbf{0} & 0 \\
0 & \textbf{0} & 0 & 0 \\
Q_{\epsilon_0} & 0 & 0 & 0
\end{pmatrix}_{\ell^2(\Z) \oplus \{0\} \oplus \{0\} \oplus \ell^2(\Z)}
\endgroup,
\end{equation}
where $\textbf{0}$ denotes the zero operator of the form $\textbf{0} : \{0\} \to \{0\},$ and where $\ell^2(\Z) \oplus \{0\} = \ker(\varGamma _\epsilon - 1)$ and $\{0\} \oplus \ell^2(\Z) = \ker(\varGamma _\epsilon + 1).$ On the other hand, the imaginary part $Q_\epsilon$ associated with $(\varGamma _\epsilon, U_\epsilon)$ admits the following off-diagonal block-operator matrix representation with respect to $\ell^2(\Z,\C^2) = \ker(\varGamma _\epsilon - 1) \oplus \ker(\varGamma _\epsilon + 1)$ as in \cref{equation: representation of R and Q};
\begin{equation}
\label{equation2: representation of Qepsilon}
Q = 
\begin{pmatrix}
0 & (Q''_{\epsilon_0})^* \\
Q''_{\epsilon_0} & 0
\end{pmatrix}_{\ker(\varGamma _\epsilon - 1) \oplus \ker(\varGamma _\epsilon + 1)}
=
\begin{pmatrix}
0 & (Q''_{\epsilon_0})^* \\
Q''_{\epsilon_0} & 0
\end{pmatrix}_{(\ell^2(\Z) \oplus \{0\}) \oplus (\{0\} \oplus \ell^2(\Z))}.
\end{equation}
It follows from \crefrange{equation1: representation of Qepsilon}{equation2: representation of Qepsilon} that $Q''_{\epsilon_0}$ is an off-diagonal block-operator matrix of the form;
\[
Q_{\epsilon_0} = 
\begingroup
\setlength\arraycolsep{4pt}
\begin{pmatrix}
0  & \textbf{0} \\
Q''_{\epsilon_0} & 0
\end{pmatrix}.
\endgroup 
\]
Since $\textbf{0}$ is a Fredholm operator of zero index, we have that $Q''_{\epsilon_0}$ is Fredholm if and only if $Q_{\epsilon_0}$ is Fredholm. In this case, we have $\ind Q''_{\epsilon_0} = \ind Q_{\epsilon_0} + \ind \textbf{0} = \ind Q_{\epsilon_0} + 0 = \ind Q_{\epsilon_0}.$ The claim follows from \cref{lemma: unitary invariance of the witten index}.
\end{proof}

It remains to compute the Fredholm index of the strictly local operator $Q_{\epsilon_0}$ given by \cref{equation1: definition of Qepsilon}, where $\theta = (\theta(x))_{x \in \Z}$ can be \textit{any} $\R$-valued sequence satisfying $q(x) = |q(x)|e^{i \theta(x)}$ for each $x \in \Z.$ Note that \cref{theorem: topological invariants of strictly local operators}~(i) is not immediately applicable to this operator $Q_{\epsilon_0},$ since it is not necessarily true that $\theta$ is convergent. More precisely, for each $\star = \pm \infty,$ if $q(\star) \neq 0,$ then we can explicitly construct $\theta$ in such a way that $\theta(\star) = \lim_{x \to \star} \theta(x)$ holds true. On the other hand, if $q(\star) =  0,$ then the same conclusion cannot be drawn in general. In order to overcome this hindrance, we shall closely follow \cite[Lemma 3.4]{Tanaka-2020};

\begin{lemma}
There exist two $\R$-valued sequences $\theta_+ = (\theta_+(x))_{x \in \Z}, \theta_- = (\theta_-(x))_{x \in \Z},$ such that  
\begin{equation}
\label{equation: modified Qepsilon}
\begin{aligned}
e^{-i \theta_+}(-2i Q_{\epsilon_0})e^{i \theta_-}
= \quad &p_+ p_+(\cdot + m) \beta(\cdot + m)e^{i(\theta - \theta_+ + \theta_-(\cdot + m))}L^{m}  \\
- &p_- p_-(\cdot - m) \beta^* e^{-i (\theta(\cdot - m) -\theta_-(\cdot - m) + \theta_+)}L^{-m} \\
- &|q|(\alpha_1 - \alpha_2(\cdot + m))e^{i(\theta_- - \theta_+)},
\end{aligned}
\end{equation} 
where the three coefficients of the above strictly local operator have the following limits for each $\star = \pm \infty:$
\begin{align}
\label{equation1: new phase}
&\lim_{x \to \star} \left( p_+(x) p_+(x + m) \beta(x + m)  e^{i (\theta(x) - \theta_+(x) + \theta_-(x + m))} \right) = 
(1 + p(\star)) b(\star)e^{i \theta(\star)}, \\
\label{equation2: new phase}
&\lim_{x \to \star} \left(p_-(x) p_-(x - m) \beta(x)^* e^{-i (\theta(x - m) - \theta_-(x - m)  + \theta_+(x))}\right) = 
(1  - p(\star)) b(\star)^* e^{-i \theta(\star)}, \\
\label{equation3: new phase}
&\lim_{x \to \star} 
\left(|q(x)|(\alpha_1(x) - \alpha_2(x + m))e^{i (\theta_-(x) - \theta_+(x))}\right) 
= 2|q(\star)| a(\star) \cosh(2 \gamma(\star)).
\end{align}
\end{lemma}
\begin{proof}
For each $x \in \Z$ we let 
\[
\star(x) := 
\begin{cases}
+\infty, & x \geq 0, \\
-\infty, & x < 0,
\end{cases}
\qquad 
\theta_{\pm}(x) := 
\begin{cases}
\theta(x), & p(\star(x)) =    \pm 1, \\
0,         & p(\star(x)) \neq \pm 1.
\end{cases}
\]
Note that \cref{equation: modified Qepsilon} immediately follows from \cref{equation1: definition of Qepsilon}. We let
\[
\Theta_1 := \theta - \theta_+ + \theta_-(\cdot + m), \quad 
\Theta_2 := \theta(\cdot - m) - \theta_-(\cdot - m)  + \theta_+, \quad 
\Theta_3 := \theta_-  - \theta_+.
\]
It suffices to prove the following equalities:
\begin{align}
\label{equation4: new phase}
&\lim_{x \to \star} \left(p_+(x) p_+(x + m) e^{i \Theta_1(x)} \right) = 
(1 + p(\star)) e^{i \theta(\star)}, \\
\label{equation5: new phase}
&\lim_{x \to \star} \left(p_-(x) p_-(x - m) e^{-i\Theta_2(x)}\right) = 
(1 - p(\star)) e^{-i \theta(\star)}, \\
\label{equation6: new phase}
&\lim_{x \to \star} 
\left(|q(x)|e^{i \Theta_3(x)}\right) = |q(\star)|.
\end{align}
Let $\star = \pm \infty,$ and let $x$ be any integer satisfying $|x| > |m|.$ If $|p(\star)| < 1,$ then $\theta_+(x) = \theta_-(x) = 0.$ In this case, \crefrange{equation4: new phase}{equation6: new phase} follow from the fact that as $x \to \star$ we have $\Theta_j(x) \to \theta(\star)$ for each $j = 1,2,$ and $\Theta_3(x) \to 0.$ On the other hand, if $|p(\star)| = 1,$ then $q(\star) = 0,$ and so \cref{equation6: new phase} becomes trivial. We need to check the following cases separately: $p(\star) = -1$ and $p(\star) = +1.$ If $p(\star) = -1,$ then \cref{equation4: new phase} holds trivially, and \cref{equation5: new phase} follows from $\theta_-(x - m) = \theta(x - m)$ and $\theta_+(x) = 0  = \theta(\star),$ where the last equality follows from \cref{equation: limits of theta and theta prime}.  Similarly, if $p(\star) = +1,$ then \cref{equation5: new phase} holds trivially, and \cref{equation4: new phase} follows from $\theta_+(x) = \theta(x)$ and $\theta_-(x+m) = 0  = \theta(\star).$
\end{proof}

Since the Fredholm index is invariant under multiplication by invertible operators, 
\[
\ind(e^{-i \theta_+} Q_{\epsilon_0} e^{i \theta_-}) = \ind Q_{\epsilon_0} = \ind(\varGamma, U).
\]
We are now in a position to apply \cref{theorem: topological invariants of strictly local operators}~(i) to $A := e^{-i \theta_+} Q_{\epsilon_0} e^{i \theta_-}.$ Since the two-sided limits of the coefficients of $-2iA_\epsilon$ are given respectively by \crefrange{equation1: new phase}{equation3: new phase}, we introduce the following notation according to \cref{equation: definition of Fourier transform of A};
\begin{align}
\label{equation2: definition of matsuzawa function} 
c(\star) &:= |q(\star)| a(\star) \cosh(2 \gamma(\star)), \\
\label{equation1: definition of matsuzawa function} 
-2if(z,\star) &:= 
(p(\star) + 1) b(\star) e^{i \theta(\star)} z^{m} + (p(\star) - 1) b(\star)^* e^{-i \theta(\star)} z^{-m}  -2c(\star),
\end{align}
where $\star = \pm \infty$ and $z \in \T.$ It follows from \cref{theorem: topological invariants of strictly local operators}~(i) that $A = e^{-i \theta_+} Q_{\epsilon_0} e^{i \theta_-}$ is Fredholm if and only if $f(\cdot,\star)$ is nowhere vanishing on $\T$ for each $\star = \pm \infty.$ In this case, we have 
\begin{equation}
\label{equation: witten index expressed as the difference of winding numbers}
\ind(\varGamma, U) = \ind A = \wn(f(\cdot, + \infty)) - \wn(f(\cdot, - \infty)),
\end{equation}
where the last equality is a special case of \cref{equation: bulk-edge correspondence}. It remains to compute the winding number of $f(\cdot, \star).$

\begin{lemma}
\label{lemma: matsuzawa function is an ellipse} 
Let $\varGamma ,C$ be as in \cref{maintheorem: generalised mko}, and let $\star = \pm \infty.$ Let $f(\cdot, \star)$ be defined by \crefrange{equation1: definition of matsuzawa function}{equation2: definition of matsuzawa function}, and let $p_\gamma(\star)$ be defined by \cref{equation: definition of pgamma}. Then the image of $\T \ni z \longmapsto f(z, \star) \in \C$ does not contain the origin if and only if $|p_\gamma(\star)| \neq |a(\star)|.$ In this case, we have
\begin{equation}
\label{equation1: winding number of matsuzawa function}
\wn(f(\cdot, \star)) = 
\begin{cases}
m \cdot \sgn p(\star), & |p_\gamma(\star)| > |a(\star)|, \\
0,                     & |p_\gamma(\star)| < |a(\star)|. \\
\end{cases}
\end{equation}
\end{lemma}
\begin{proof}
Let us first prove that the image of $\T \ni z \longmapsto f(z, \star) \in \C$ does not contain the origin if and only if $|p(\star) b(\star)| \neq |c(\star)|,$ and 
\begin{equation}
\label{equation2: winding number of matsuzawa function}
\wn(f(\cdot, \star)) = 
\begin{cases}
m \cdot \sgn p(\star), & |p(\star) b(\star)| > |c(\star)|, \\
0,                     & |p(\star) b(\star)| < |c(\star)|. \\
\end{cases}
\end{equation}
Let us consider the following function on $\R;$
\begin{align*}
2F(s) 
&:= (|p(\star)b(\star)| + |b(\star)|) e^{is} + (|p(\star)b(\star)| - |b(\star)|)e^{-i s} \\
&= 2 |p(\star)b(\star)| \cos s + i 2|b(\star)| \sin s, \qquad s \in \R.
\end{align*}
Since $p(\star) = \sgn p(\star)|p(\star)|$ and $b(\star) = e^{i \theta'(\star)}|b(\star)|,$ for each $t \in [0,2\pi]$ we have
\begin{align*}
-2if(e^{i t}, \star) + 2c(\star) 
&= (p(\star) + 1) b(\star) e^{i \theta(\star)} e^{i mt} + (p(\star) - 1) b(\star)^* e^{-i \theta(\star)} e^{-i mt} \\
&= \sgn p(\star) \cdot 2F(\sgn p(\star)(\theta(\star) + \theta'(\star) + mt)).
\end{align*}
It follows that $-if(e^{i t}, \star) = \sgn p(\star) \cdot F(\sgn p(\star)(\theta(\star) + \theta'(\star) + mt)) -c(\star)$ for each $t \in [0,2\pi],$ where the constant $-i$ does not play any significant role in this proof.  If $p(\star)b(\star) = 0,$ then the image of the function $[0,2\pi] \ni t \longmapsto -if(e^{i t},\star) \in \C$ coincides with that of the vertical line segment $[-1,1] \ni t \longmapsto -c(\star) + i t|b(\star)| \in \C$ passing through $-c(\star).$ That is, the image of $f(\cdot, \star)$ does not contain the origin if and only if $|c(\star)| \neq 0 = |p(\star)b(\star)|,$ and in this case $\wn(f(\cdot, \star)) = 0.$ This is a special case of \cref{equation2: winding number of matsuzawa function}.

On the other hand, if $p(\star)b(\star) \neq 0,$ then the image of the curve $[0,2\pi] \ni t \longmapsto -if(e^{i t},\star) \in \C$ is the ellipse in \cref{figure: ellipse} with $m \cdot \sgn p(\star)$ being its winding number with respect to the center $-c(\star)$ on the real axis;

\begin{figure}[H]
\centering
\begin{tikzpicture}
\begin{axis}[axis y line=none,ticks=none, xmin= -8, xmax=8, ymin=-2, ymax=2, legend pos = north west, axis lines=center, xlabel=$\Re$, xlabel style={anchor = west}
, width = \textwidth, height = 0.5\textwidth
]
	\addplot [domain=-2*pi:2*pi,samples=50, smooth]({3*cos(deg(x))},{sin(deg(x))}); 
	\addplot [mark=none,forget plot, dashed] coordinates {(3, -1.5) (3, 1.5)};
	\addplot [mark=none,forget plot, dashed] coordinates {(-3, -1.5) (-3, 1.5)};
	\addplot [mark=none,forget plot, dashed] coordinates {(0, -1.5) (0, 1.5)};
	\draw [fill] (-3,0) circle (1.5 pt) node [anchor = north east] {$-c(\star) - |p(\star)b(\star)|$};
	\draw [fill] (3,0) circle (1.5 pt) node [anchor = north west] {$-c(\star) + |p(\star)b(\star)|$};
	\draw [fill] (0,0) circle (1.5 pt) node [anchor = north west] {$-c(\star)$};
\end{axis}
\end{tikzpicture}
\caption{The above figure shows the image of the curve $[0,2\pi] \ni t \longmapsto -if(e^{i t},\star) \in \C.$}
\label{figure: ellipse}
\end{figure}
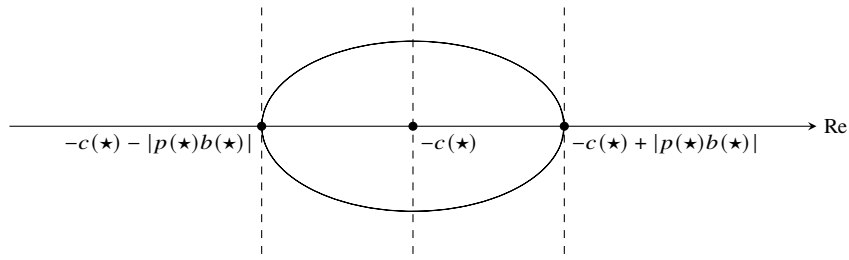

If $|p(\star)b(\star)| > |c(\star)|,$ then the origin is inside the interior of the ellipse, and so $\wn(f(\cdot,\star)) = \wn(-if(\cdot,\star)) =   \sgn p(\star).$ If $|p(\star)b(\star)| < |c(\star)|,$ then the origin is inside the exterior of the ellipse, and so $\wn(f(\cdot,\star)) = 0.$ Clearly, the ellipse $-if$ goes through the origin if and only if $|p(\star)b(\star)| = |c(\star)|.$

It remains to check that \cref{equation1: winding number of matsuzawa function} coincides with \cref{equation2: winding number of matsuzawa function}. If the notation $\lessgtr$ simultaneously denotes $>, =, <,$ then
$|p(\star) b(\star)| \lessgtr |c(\star)|$ if and only if $p(\star)^2 (1-a(\star)^2) \lessgtr |q(\star)|^2 a(\star)^2 \cosh^2(2 \gamma(\star))$ if and only if $p(\star)^2 \lessgtr a(\star)^2(p(\star)^2 + |q(\star)|^2\cosh^2(2 \gamma(\star)).$ Rearranging the last expression gives $|p_\gamma(\star)| \lessgtr |a(\star)|.$ The claim follows.
\end{proof}

\begin{proof}[Proof of \cref{maintheorem: generalised mko}~(i)]
The index formula \cref{equation: Witten index formula} immediately follows from \cref{equation: witten index expressed as the difference of winding numbers} and \cref{equation1: winding number of matsuzawa function}.
\end{proof}

It might be possible to give another proof for the index formula \cref{equation: Witten index formula} by making use of the recent developments of the scattering-theoretic techniques for discrete-time quantum walks \cite{Suzuki-2016,Richard-Suzuki-Tiedra-2017,Richard-Suzuki-Tiedra-2018,Maeda-Sasaki-Segawa-Suzuki-Suzuki-2018a,Morioka-2019,Wada-2020}. This possibility is briefly mentioned in \cite[\textsection 6]{Suzuki-Tanaka-2019}.

\subsubsection{Proof of \texorpdfstring{\cref{maintheorem: generalised mko}~(ii)}{Theorem C (ii)}}

\begin{proof}[Proof of \cref{maintheorem: generalised mko}~(ii)]
Note first that $U_m$ is a strictly local operator of the following form;
\[
U_m = 
\begin{pmatrix}
p e^{-2 \gamma(\cdot + 1)} a + q L^m e^{\gamma - \gamma(\cdot + 1)} b &   pe^{\gamma - \gamma(\cdot + 1)} b^* - q  L^m e^{2 \gamma} a \\
L^{-m} q^*  e^{-2 \gamma(\cdot + 1)} a - p(\cdot - m) e^{\gamma - \gamma(\cdot + 1)} b  &   L^{-m}q^* e^{\gamma - \gamma(\cdot + 1)} b^*  + p(\cdot - m)e^{2 \gamma}  a
\end{pmatrix}.
\]
It follows from \cref{theorem: topological invariants of strictly local operators}~(ii) that
\begin{align*}
&\ess(U_m) = \ess(U_m(- \infty)) \cup \ess(U_m(+ \infty)), \\
&\ess(U_m(\star)) = \bigcup_{z \in \T} \sigma\left(\hat{U}_m(z,\star)\right), \qquad \star = \pm \infty,
\end{align*}
where for each $\star = \pm \infty$ and each $z \in \T$ the $2 \times 2$ matrices $U_m(\star)$ and $\hat{U}_m(z,\star)$ are defined respectively by:
\begin{align*}
U_m(\star) &:= 
\begin{pmatrix}
q(\star) b(\star)L^m  + p(\star) a(\star) e^{-2 \gamma(\star)}   &  - (q(\star) a(\star) e^{2 \gamma(\star)} L^m -  p(\star) b(\star)^*)  \\
q(\star)^*a(\star) e^{-2 \gamma(\star)} L^{-m} - p(\star) b(\star)  &   q(\star)^* b(\star)^*L^{-m}  + p(\star) a(\star)e^{2 \gamma(\star)}
\end{pmatrix}, \\
\hat{U}_m(z,\star) 
&:= 
\begin{pmatrix}
q(\star) b(\star)z^m  + p(\star) a(\star) e^{-2 \gamma(\star)}   &  - (q(\star) a(\star) e^{2 \gamma(\star)} z^m -  p(\star) b(\star)^*)  \\
q(\star)^*a(\star) e^{-2 \gamma(\star)} z^{-m} - p(\star) b(\star)  &   q(\star)^* b(\star)^*z^{-m}  + p(\star) a(\star)e^{2 \gamma(\star)}
\end{pmatrix}.
\end{align*}
Let $\star = \pm \infty$ be fixed. It remains to compute $\sigma'(\star) := \bigcup_{t \in [0,2\pi]} \sigma\left(\hat{U}_m(e^{it},\star)\right).$ We let
\[
\hat{U}_m(e^{it},\star) =: 
\begin{pmatrix}
X_1(e^{it}) & -Y_1(e^{it}) \\
Y_2(e^{it}) & X_2(e^{it})
\end{pmatrix}, \qquad t \in [0,2\pi].
\]
We get the following characteristic equation;
\begin{equation}
\label{equation1: characteristic equation}
\det(\hat{U}_m(e^{it},\star) - \lambda) = \lambda^2 -(X_1(e^{it}) + X_2(e^{it}))\lambda + X_1(e^{it})X_2(e^{it}) + Y_1(e^{it})Y_2(e^{it}).
\end{equation}
Since the produce $q(\star) b(\star)$ can be written as $q(\star) b(\star) = |q(\star) b(\star)|e^{i (\theta(\star) + \theta'(\star))}$ by \cref{equation: limits of theta and theta prime}, we obtain the following two equalities:
\begin{align*}
&X_1(e^{it}) + X_2(e^{it})
= 2|q(\star)  b(\star)| \cos(\theta(\star) + \theta'(\star) + mt) + 2 p(\star) a(\star) \cosh(2\gamma(\star)), \\
&X_1(e^{it})X_2(e^{it}) + Y_1(e^{it})Y_2(e^{it}) = 1.
\end{align*}
Then the characteristic equation \cref{equation1: characteristic equation} becomes
\begin{equation}
\label{equation2: characteristic equation}
\lambda^2 - 2(p(\star) a(\star) \cosh(2\gamma(\star)) + |q(\star)  b(\star)| \cos(\theta(\star) + \theta'(\star) + mt))\lambda + 1 = 0.
\end{equation}
This equation motivates us to introduce the following notation;
\begin{align*}
\Lambda(\star, s) &:= p(\star) a(\star)\cosh(2\gamma(\star))  + |q(\star) b(\star)| s, &&-1 \leq s \leq 1, \\
\lambda_{\pm}(\star,s) &:= \Lambda(\star,s) \pm \sqrt{\Lambda(\star,s)^2 - 1},  &&-1 \leq s \leq 1, 
\end{align*}
Indeed,  \cref{equation2: characteristic equation} becomes $\lambda^2 - 2 \Lambda(\star,\cos(\theta(\star) + \theta'(\star) + mt)) \lambda + 1 = 0$ with the above notation, and so $\sigma\left(\hat{U}_m(e^{it},\star)\right)$ is a finite set consisting only of $\lambda_{\pm}(\star,\cos(\theta(\star) + \theta'(\star) + mt))$ for each $t \in [0,2\pi].$ We have
\[
\sigma'(\star) 
= \bigcup_{t \in [0,2\pi]} \sigma\left(\hat{U}_m(e^{it},\star)\right)
=  \bigcup_{s \in [-1,1]} \{\lambda_{\pm}(\star,s)\}
=  \bigcup_{s \in [-1,1]} \{\lambda_{+}(\star,s)^{\pm 1}\},
\]
where the second equality follows from the fact that $[0,2\pi] \ni t \longmapsto \cos(\theta(\star) + \theta'(\star) + mt) \in [-1,1]$ is surjective and the last equality follows from $\lambda_{+}(\star,t)\lambda_{-}(\star,t) = 1$ for each $t \in [0,2\pi].$ It follows that $\sigma'(\star)$ coincides with the set $\sigma(\star)$ given by \cref{equation: definition of sigma star}. Note first that $[\Lambda_-(\star),\Lambda_+(\star)] \subseteq [-1,\infty)$ follows from
\[
-1 \leq -|q(\star)b(\star)| \leq |a(\star)b(\star)|\cosh(2 \gamma(\star)) -|q(\star)b(\star)| = \Lambda_-(\star) \leq \Lambda_+(\star).
\]
If $p(\star)a(\star) = 0,$ then $\Lambda_+(\star) = |q(\star)b(\star)| \leq 1.$ This is a special case of Case I, since $\gamma_-(\star) = \gamma_+(\star) = \infty$ according to \cref{equation: definition of gammaj}. It remains to consider the case $p(\star)a(\star) \neq 0.$ We shall make use of the fact that the hyperbolic cosine is an even function throughout. It follows from \cref{equation: definition of gammaj} that
\begin{equation}
\label{equation: paqb}
|p(\star)a(\star)|\cosh(2 \gamma_\pm(\star)) = 1 \pm |q(\star)b(\star)|.
\end{equation}

\textnormal{\textbf{Case I. }} If $|\gamma(\star)| \leq \gamma_-(\star),$ then
\[
\Lambda_+(\star) 
\leq |p(\star) a(\star)|\cosh(2\gamma_-(\star))  + |q(\star) b(\star)| = 1,
\]
where the first inequality follows from $\cosh(2\gamma(\star)) \leq \cosh(2\gamma_-(\star))$ and the last equality follows from \cref{equation: paqb}. Thus $[\Lambda_-(\star),\Lambda_+(\star)] \subseteq  [-1,1].$ 

\textnormal{\textbf{Case II. }} If $\gamma_-(\star) < |\gamma(\star)| < \gamma_+(\star),$ then it follows from \cref{equation: paqb} that
$
\Lambda_-(\star) < 1 < \Lambda_+(\star).
$
It follows that the interval $[\Lambda_-(\star),\Lambda_+(\star)] \subseteq [-1,\infty)$ can be written as the following union;
\[
[\Lambda_-(\star),\Lambda_+(\star)] = [\Lambda_-(\star),1] \cup [1,\Lambda_+(\star)].
\]

\textnormal{\textbf{Case III. }} If $\gamma_+(\star) \leq |\gamma(\star)|.$ Then $[\Lambda_-(\star),\Lambda_+(\star)] \subseteq  [1,\infty)$ follows from 
\[
1 = |p(\star)a(\star)|\cosh(2 \gamma_+(\star)) - |q(\star)b(\star)| \leq \Lambda_-(\star),
\]
where the first equality follows from \cref{equation: paqb} and the last inequality follows from $\cosh(2\gamma_+(\star)) \leq \cosh(2\gamma(\star)).$ 
\end{proof}

In the setting of $2$-phase quantum walks, a typical computation of the essential spectrum makes use of the discrete Fourier transform and Weyl's criterion for the essential spectrum (see, for example, \cite[Lemma 3.3]{Fuda-Funakawa-Suzuki-2017}). Weyl's criterion is applicable to, for example, non-compact perturbations (see, for example, \cite{Sasaki-Suzuki-2017}), but its usage is restricted to normal operators. This is why Weyl's criterion is not suitable for \cref{maintheorem: generalised mko}~(ii).

\section{Conclusion}
\label{section: concluding remarks}

\subsection{Summary}

The following is a brief summary of the present article. A chiral pair on a Hilbert space $\cH$ is by definition any pair $(\varGamma,U)$ of a unitary self-adjoint operator $\varGamma : \cH \to \cH$ and a bounded operator $U : \cH \to \cH$ satisfying the chiral symmetry condition \cref{equation: chiral symmetry}. It is shown in \cref{section: preliminaries} that we can assign to each abstract chiral pair $(\varGamma, U)$ the well-defined Witten index, denoted by $\ind(\varGamma,U)$ in this paper. Note that this assignation of the Fredholm index is a natural generalisation of the existing index theory \cite{Cedzich-Grunbaum-Stahl-Velazquez-Werner-Werner-2016,Cedzich-Geib-Grunbaum-Stahl-Velazquez-Werner-Werner-2018,Cedzich-Geib-Stahl-Velazquez-Werner-Werner-2018,Suzuki-2019,Suzuki-Tanaka-2019,Matsuzawa-2020} for essentially unitary $U,$ where $\ind(\varGamma,U)$ is referred to as the \textbi{symmetry index} in the first three papers.

A motivating example for this paper is the non-unitary time-evolution $U_{\textnormal{mko}}$ defined by \cref{equation: definition of evolution operator of MKO}, where we assume the existence of the two-sided limits as in \cref{equation: existence of limits}. Recall that this evolution operator is consistent with the experimental setup in \cite{Regensburger-Bersch-Miri-Onishchukov-Christodoulides-Peschel-2012}. It is shown in \cref{theorem: MKO model} that the operator $U_{\textnormal{mko}}$ forms a chiral pair with respect to the unitary self-adjoint operator $\varGamma_{\textrm{mko}} := (\sigma_2 C_1 S \sigma_2) \varGamma_2 (\sigma_2 C_1 S \sigma_2)^*,$ where $\sigma_2$ denotes the second Pauli matrix, and that the chiral pair $(\varGamma_{\textnormal{mko}}, U_{\textnormal{mko}})$ can be naturally generalised to another chiral pair $(\varGamma_m, U_m),$ where $m$ can be any fixed integer. This new model $(\varGamma_m, U_m)$ also unifies several one-dimensional unitary quantum walk models as in \cref{section: chirally symmetric quantum walks}. Complete classification of the two associated topological invariants $\ind(\varGamma_m, U_m)$ and $\ess(U_m)$ can be collectively found in \cref{maintheorem: generalised mko}. Our classification of $\ind(\varGamma_m, U_m)$ makes use of an abstract form of the one-dimensional bulk-boundary correspondence, the precise statement of which can be found in \cref{theorem: topological invariants of strictly local operators}~(i).

Finally, it is shown in \cref{lemma: essential spectrum of essentially unitary U} that given an abstract chiral pair $(\varGamma,U)$ with $U$ being essentially unitary, we have that $\ind(\varGamma,U)$ is a well-defined Fredholm index if and only if $U$ is essentially gapped in the sense that $-1,+1 \notin \ess(U).$ It turns out that this characterisation does not hold true in general, if $U$ fails to be essentially unitary. To put this into context, we consider the non-unitary evolution $U_{\textnormal{mko}}.$ It is shown in \cref{example: second example} that we can choose the asymptotic values $\theta_1(\pm \infty), \theta_2(\pm \infty), \gamma(\pm \infty),$ in such a way that $U_{\textnormal{mko}}$ is essentially gapless, yet $\ind(\varGamma,U)$ is a well-defined non-zero integer.  

\subsection{Discussion}
The main results of the current paper may stimulate further developments in the rigorous mathematical studies of non-unitary discrete-time quantum walks. In particular, each of the following specific topics is the subject of another paper in preparation.

\subsubsection{Further spectral analysis of the Mochizuki-Kim-Obuse model}
Complete classification of $\ess(U_{\textnormal{mko}})$ is given in this paper. In particular, we show that $\ess(U_{\textnormal{mko}})$ is a subset of $\T \cup \R,$ and that it depends only on the asymptotic values $\theta_1(\pm \infty), \theta_2(\pm \infty), \gamma(\pm \infty).$ Note, however, that it is not known to the authors whether or not the entire spectrum of $U_{\textnormal{mko}}$ is also a subset of $\T \cup \R.$ Detailed spectral analysis of the evolution-operator $U_{\textnormal{mko}}$ may turn out to be difficult, partly because the discrete spectrum of such a non-normal operator is in general laborious to characterise (see, for example, \cite[\textsection III]{Boussaid-Comech-2019}). Note also that we expect the discrete spectrum to be non-stable under compact perturbations unlike $\ess(U_{\textnormal{mko}}).$ 

\subsubsection{Topologically protected bound states}

Let $(\varGamma,U)$ be a chiral pair. If $U$ is unitary, then the non-zero vectors in $\ker(U \mp 1)$ can be referred to as \textbi{topologically protected bound states} \cite{Kitagawa-Rudner-Berg-Demler-2010,Kitagawa-Broome-Fedrizzi-Rudner-Berg-Kassal-Aspuru-Demler-White-2012,Suzuki-2019,Suzuki-Tanaka-2019,Matsuzawa-2020}. It is well-known that the Witten index $\ind(\varGamma,U)$ gives a lower bound for the number of  topologically protected bound states in the following precise sense (see, for example, \cite[Theorem 3.4~(ii)]{Suzuki-2019});
\begin{equation}
\label{equation: topologically protected bound states}
|\ind(\varGamma,U)| \leq \dim \ker(U - 1) + \dim \ker(U + 1),
\end{equation}
where the chiral pair $(\varGamma,U)$ is assumed to be Fredholm. It follows that if $\ind(\varGamma,U)$ is non-zero, then $U$ has at least one topologically protected bound state. Whether or not an estimate analogous to \cref{equation: topologically protected bound states} holds true for non-unitary $U$ is an open problem.

\begin{acknowledgements}
The authors are deeply indebted to the members of the Shinshu Mathematical Physics Group for extremely valuable discussions and comments. Our sincerely thanks go to T.~Daniels for carefully reading the preprint. Y.~T.~was supported by JSPS KAKENHI Grant Number 20J22684. This work was partially supported by the Research Institute for Mathematical Sciences, a Joint Usage/Research Center located in Kyoto University.
\end{acknowledgements}

\bibliographystyle{alpha}
\bibliography{Bibliography}
\end{document}